\documentclass[aps, physrev, reprint,superscriptaddress,nofootinbib, showkeys]{revtex4-2}

\usepackage{float}
\makeatletter
\let\newfloat\newfloat@ltx
\makeatother

\usepackage[english]{babel}

\usepackage[utf8]{inputenc}
\usepackage{selinput}
\usepackage[normalem]{ulem}
\usepackage[shortlabels]{enumitem}

\usepackage{braket}
\usepackage{amsthm}
\usepackage{amsmath}
\usepackage{mathtools}
\usepackage{physics}
\usepackage{xcolor}
\usepackage{graphicx}
\usepackage{adjustbox}
\usepackage{placeins}
\usepackage[T1]{fontenc}
\usepackage{lipsum}
\usepackage{csquotes}
\usepackage{bm}

\usepackage[linesnumbered,ruled,vlined]{algorithm2e}
\SetKwInput{kwInit}{Init}

\def\Cbb{\mathbb{C}}

\def\HC{\mathcal{H}}

\def\LC{\mathcal{L}}

\def\ad{^{\dagger}}

\newcommand{\old}[1]{}

\usepackage[toc,page]{appendix}
\usepackage[colorlinks=true,citecolor=blue,linkcolor=magenta]{hyperref}

\usepackage{bm}
\usepackage{bbm}
\usepackage{comment}
\usepackage{mathdots}
\usepackage{verbatim}
\usepackage{natbib}
\usepackage{nccmath}
\usepackage{amsfonts} 

\newcommand{\ketbraq}[1]{\ketbra{#1}{#1}}
\newcommand{\bramatket}[3]{\langle #1 \hspace{1pt} | #2 | \hspace{1pt} #3 \rangle}

\newcommand{\avg}[1]{\langle #1\rangle }
\newcommand{\dya}[1]{\ket{#1}\!\bra{#1}}

\newcommand{\poly}{\operatorname{poly}}

\newcommand{\Ebb}{\mathbb{E}}

\newcommand{\Ubb}{\mathbb{U}}

\newcommand{\BC}{\mathcal{B}}

\newcommand{\GC}{\mathcal{G}}

\newcommand{\MC}{\mathcal{M}}
\newcommand{\NC}{\mathcal{N}}
\newcommand{\OC}{\mathcal{O}}

\newcommand{\TC}{\mathcal{T}}

\newcommand{\VC}{\mathcal{V}}
\newcommand{\WC}{\mathcal{W}}

\newcommand{\Var}{{\rm Var}}

\renewcommand{\leq}{\leqslant}

\newcommand{\spn}{{\rm span}}

\renewcommand{\vec}[1]{\boldsymbol{#1}}  

\newcommand{\rholh}{\hat{\rho}}

\def\Hom{{\rm Hom}}

\newcommand{\al}{\alpha }
\newcommand{\bt}{\theta }
\newcommand{\gm}{\gamma }

\newcommand{\dl}{\delta }

\newcommand{\ep}{\epsilon}

\newcommand{\lm}{\lambda }
\newcommand{\Lm}{\Lambda }

\newcommand{\sg}{\sigma }

\newcommand{\om}{\omega }

\def\eye{\mathds{1}}

\def\U{\mathrm{U}}

\def\l{\lambda}

\newtheorem{theorem}{Theorem}
\newtheorem{lemma}{Lemma}

\newtheorem{corollary}{Corollary}

\newtheorem{definition}{Definition}

\usepackage{amssymb}
\usepackage{dsfont}

\newsavebox{\mstrut}

\newcommand{\rrangle}{%
    \rangle\kern-0.4\ht\mstrut\right\rangle%
}
\newcommand{\llangle}{%
    \langle\kern-0.4\ht\mstrut\left\langle%
}

\newcommand{\kett}[1]{\sbox{\mstrut}{$#1$}%
    \mathinner{\left.\left|{#1}\right\rrangle} }

\newcommand{\QST}{U_{\rm QST}}
\newcommand{\QS}{\rm Q.S.}
\newcommand{\CT}{\rm C.T.}
\newcommand{\QT}{\rm Q.T.}

\usepackage[caption = false]{subfig}
\usepackage{ytableau}
\usepackage{multirow}

\usepackage{algpseudocode}

\usepackage{tabularx,booktabs}
\newcolumntype{L}{>{$}l<{$}}  

\usepackage{titlesec}
\titleformat{\paragraph}[runin]
  {\normalfont\bfseries}
  {}
  {0pt}
  {} 

\begin{document}

\title{Practical framework for simulating permutation-equivariant quantum circuits}
\author{Su Yeon Chang}
\thanks{suyeon.chang97@gmail.com}
\affiliation{Theoretical Division, Los Alamos National Laboratory, Los Alamos, NM 87545, USA}

\author{Martin Larocca }
\affiliation{Theoretical Division, Los Alamos National Laboratory, Los Alamos, NM 87545, USA}
\affiliation{Quantum Science Center,  Oak Ridge, TN 37931, USA}

\author{M. Cerezo}
\thanks{cerezo@lanl.gov}
\affiliation{Information Sciences, Los Alamos National Laboratory, Los Alamos, NM 87545, USA}
\affiliation{Quantum Science Center,  Oak Ridge, TN 37931, USA}

\begin{abstract}

Understanding which subclasses of quantum circuits are efficiently classically simulable is fundamental to delineating the boundary between classical and quantum computation. In this context, it is well known that certain tasks based on permutation-equivariant unitaries--i.e., $n$-qubit circuits whose action commutes with the qubit-permuting representation of the symmetric group $S_n$--can be simulated in polynomial time. However, existing approaches scale as $\OC(n^7)$, and can rapidly become prohibitively expensive. In this work, we introduce a practical algorithm for simulating $S_n$-equivariant circuits under the assumption that the gate generators are at most $k$-local, with $k\in\mathcal{O}(1)$. The resulting method runs in $\mathcal{O}(n^{\omega+1})$ time for constant depth, where $\omega$ is the matrix multiplication exponent, significantly lowering the polynomial degree compared to existing techniques. Finally, we numerically validate this scaling by simulating the dynamical evolution of the Lipkin--Meshkov--Glick model, and show that for $n=512$ spins, a standard laptop can compute the concurrence of the evolved state in under two minutes.
\end{abstract}

\maketitle

\section{Introduction}

Efficient classical simulation of quantum circuits typically arises when the evolution is confined to an effectively low-dimensional subspace, so that the relevant degrees of freedom scale only polynomially with system size~\cite{ermakov2024unified,feng2025quon,cerezo2023does,cirstoiu2024fourier,teng2025leveraging,angrisani2024classically,angrisani2025simulating}.
This viewpoint unifies several of the best-understood simulable families: low-entanglement dynamics admits tensor-network descriptions with controlled bond dimension~\cite{Vidal2003Efficient}, $n$-qubit Clifford circuits correspond to linear symplectic updates of Pauli operators via a binary $\mathrm{Sp}(2n,\mathbb{Z}_2)$ action~\cite{gottesman1997stabilizer,gottesman1998heisenbergrepresentation,bravyi2016improved}, and matchgate/fermionic-linear-optics circuits act linearly on Majorana modes, yielding a representation of $\mathrm{SO}(2n)$~\cite{knill2001fermionic,jozsa2008matchgates,bravyi2002fermionic,brod2016efficient}.
Recent work has sought additional simulable families, both to sharpen the boundary between classical and quantum computation and to provide principled baselines for benchmarking quantum devices.
However, Lie-algebraic classification results indicate that such tractable cases are rare, as generic local gates generate quasi-universal dynamics with no polynomially-sized subspaces, thus precluding their classical simulation~\cite{ragone2023lie,wiersema2023classification, kokcu2024classification,kazi2022landscape}. This makes known simulable classes even more precious and rare, and further motivates improving surrogate simulation techniques for circuit classes that do possess underlying low-dimensional structures.

Another notable class of efficiently simulable circuits is that of $S_n$-equivariant (permutation-equivariant) evolutions~\cite{schatzki2022theoretical,kazi2023universality}, i.e., $n$-qubit unitaries that commute with the qubit-permuting representation of the symmetric group $S_n$.
For these circuits, Ref.~\cite{anschuetz2022efficient} established a polynomial-time classical simulation algorithm via a tensor-network contraction scheme whose worst-case scaling is $\mathcal{O}(n^7)$.
Despite the fact that this is a polynomial scaling---and thus allows us to view $S_n$-equivariant circuits as ``efficiently'' simulable in principle---the associated algorithms can still be prohibitively expensive at relatively modest problem sizes.
Notably, although such circuits have clear practical uses for simulating settings where labels carry no intrinsic meaning---such as models of indistinguishable particles, collective-spin Hamiltonians, and equivariant architectures for graph-structured (geometric) quantum machine learning~\cite{schatzki2022theoretical,meyer2023exploiting,verdon2019quantumgraph,skolik2022equivariant,mernyei2022equivariant,east2023all,gibbs2024exploiting,yadin2023thermodynamics,park2026hyqurp, kraus2013ground,chang2025primer}---there has been surprisingly little effort devoted to reducing their simulation cost.

In this work, we take a step in this direction by introducing an efficient and practical classical simulation framework for $S_n$-equivariant quantum circuits based on Schur--Weyl block decomposition and simple matrix multiplication, achieving an improved worst-case complexity of $\mathcal{O}(n^4)$.
We also present a numerical routine for computing the block-diagonal form of arbitrary $k$-local $S_n$-equivariant Pauli operators, for $k\in\mathcal{O}(1)$, with complexity $\mathcal{O}(n^2)$.
Beyond circuit-level simulation, we introduce a hybrid scheme that incorporates permutation-invariant classical shadows (PI-CS)~\cite{sauvage2024classical}, allowing one to collect data from a generic input quantum state in an initial measurement phase and then compute expectation values with our proposed techniques.
We demonstrate these methods with an end-to-end numerical study of the Lipkin--Meshkov--Glick model~\cite{lipkin1965validity,dusuel2004finite, dusuel2005continuous, vidal2006concurrence}, including state preparation, time evolution, and observable estimation to compute the spin concurrence~\cite{wootters1998entanglement}, a central quantity in quantum information theory.
Finally, leveraging the resulting large-scale simulations, we numerically investigate the convergence of the concurrence to a proposed thermodynamic limit~\cite{vidal2006concurrence,pal2023complexity}.

The paper is structured as follows. We start with a brief summary on the theory of permutation-equivariant quantum circuits in Section~\ref{sec:permutation_equivariant_circuit}, with a special focus on their underlying representation-theoretical properties. Then, in Section~\ref{sec:classical_simulation}, we introduce a classical simulation method for such circuits,  provided that their gates are generated by at most two-local Pauli operators. We also extend the study to arbitrary $k$-local Pauli generators, and present the time complexity required to evaluate their matrix expression in the Schur basis. In Section~\ref{sec:numerics}, we showcase our method to classically simulate properties of the Lipkin--Meshkov--Glick (LMG) model.  Our discussions are presented in Section~\ref{sec:discussion}.

\section{Preliminaries\label{sec:permutation_equivariant_circuit}}

The main goal of this work is to develop a new algorithm to efficiently classically simulate permutation-equivariant quantum circuits. Before proceeding to our main results, we find it convenient to first define a few important concepts.  

\subsection{$S_n$-equivariant operators}

We begin by denoting the $n$-qubit Hilbert space  as $\HC=(\mathbb{C}^2)^{\otimes n}$, $\mathbb{U}(\HC)$ the unitary group, and $\LC(\HC)=\HC\otimes \HC^*$ the space of linear operators acting on $\HC$. Then, let $S_n$ denote the symmetric group, i.e., the group whose elements are all the bijections from the set of $n$ elements onto itself, and $R:S_n\rightarrow \mathbb{U}(\HC)$  the so-called qubit-permuting representation. That is, given $\sg \in S_n$, we have 
\begin{equation}
    R(\sg) \bigotimes_{i=1}^n\ket{\psi_i} = \bigotimes_{i=1}^n \ket{\psi_{\sg^{-1}(i)}}\;. 
    \label{eq:sn_representation}
\end{equation}

From here, we define $S_n$-equivariant operators~\cite{meyer2023exploiting, ragone2022representation, schatzki2022theoretical, nguyen2022atheory} as follows
\footnote{A note on terminology. \emph{Equivariance} is a property of a map between $G$-modules: a linear map $\Phi:V\to W$ is $G$-equivariant if $\Phi\circ \rho_V(g)=\rho_W(g)\circ \Phi$ for all $g\in G$. Here $\rho_V$ and $\rho_W$ are group homomorphisms from $G$ to $\U(V)$ and $\U(W)$ respectively. In contrast, \emph{invariance} is a property of a vector $v$ in a $G$-module $V$: it is $G$-invariant if $\rho_V(g)v=v$ for all $g\in G$. Given that $\Hom(V,W)$ is itself a group module (one isomorphic to $V^* \otimes W$), a $G$-equivariant operator $A\in \Hom(V,W)$ corresponds precisely to a $G$-invariant vector in that $G$-module. For this reason, the literature sometimes uses ``$S_n$-invariant operator'' and ``$S_n$-equivariant operator'' interchangeably.}.

\begin{definition}[$S_n$-equivariant operator]\label{def:equivariant}
Let $R$ be the representation of $S_n$ defined in Eq.~\eqref{eq:sn_representation}. An operator $U\in \LC(\HC)$ is called $S_n$-equivariant if and only if $U\in \mathrm{comm}(S_n)$, where $\mathrm{comm}(S_n)$ denotes the (first-order) commutant of $R(S_n)$, defined as
\begin{equation}\label{eq:comm}
  \mathrm{comm}(S_n)=\{A\in \LC(\HC)\mid [R(g),A]=0,\ \forall g\in S_n\}\,.
\end{equation}
In particular, if $U$ is generated by a Hermitian operator $H$ via $U(t)=e^{-itH}$, then $U(t)$ is $S_n$-equivariant for all $t\in\mathbb R$ if and only if $H\in \mathrm{comm}(S_n)$.
\end{definition}

In particular, we can obtain a basis of $\mathrm{comm}(S_n)$ by symmetrizing a basis of $\LC(\HC)$, i.e., by twirling over $S_n$ all Pauli strings. Here, given an operator $A\in\LC(\HC)$, we define its twirl $\TC_{S_n}(A)$ as 
\begin{equation}
\TC_{S_n}(A)= \frac{1}{n!}\sum_{\sigma \in S_n} R(\sigma) A R(\sigma)\ad\,.
\end{equation}
From the previous, one can readily see that the basis of $\mathrm{comm}(S_n)$ will be given by the sum of all distinct Pauli strings that have $k_X$ $X$ symbols, $k_Y$
$Y$ symbols, and $k_Z$ $Z$ symbols~\cite{kazi2023universality}. Specifically, defining the operator
\begin{equation}
P_{\vec{k}} = X^{\otimes k_X} Y^{\otimes k_Y} Z^{\otimes k_Z}\eye^{\otimes n - k},
\label{eq:Pk}
\end{equation}
with weight vector $\vec{k} = (k_X, k_Y, k_Z)$ and  total locality $k = \abs{\vec{k}} = k_X + k_Y + k_Z$, then we have
\begin{equation}
    \mathrm{comm}(S_n)={\rm span}_{\mathbb{C}}\{\TC_{S_n}(P_{\vec{k}})\}_{\vec{k}}\,.
\end{equation}
Here, we refer to $\TC_{S_n}(P_{\vec{k}})$ as a symmetrized Pauli string.

Note that for fixed $k$ the only freedom is on the triplet $(k_X, k_Y, k_Z)$. Hence, the number of nonnegative integer solutions to $k =  k_X + k_Y + k_Z$
is the stars-and-bars count
\begin{equation}
    N(k)=\binom{k+2}{2}\,.
\end{equation}
Thus, summing $k$ from $0$ to $n$ leads to 
\begin{equation}\label{eq:dim-com}
    \dim(\mathrm{comm}(S_n))=\sum_{k=0}^nN(k)=\binom{n+3}{3}={\rm Te}_{n+1}\,,
\end{equation}
where ${\rm Te}_{x}$ denotes the $x$-th Tetrahedral number. 

In what follows, we will consider the task of classically simulating an expectation value of the form
\begin{equation}\label{eq:expec}
    f(\rho)=\Tr[U\rho U\ad O]\,,
\end{equation}
where $U$ is an $S_n$-equivariant unitary, and $O$ an $S_n$-equivariant Hermitian measurement operator. In turn, we will assume that $U$ is a unitary quantum circuit composed of $L$ gates as 
\begin{equation}
    U = \prod_{\ell = 1}^L  e^{-i H_\ell}\,,   \label{eq:U_expression}
\end{equation}
where $H_\ell$ are $S_n$-equivariant Hermitian gate generators. This general form encompasses parametrized circuits in variational quantum algorithms and quantum machine learning~\cite{cerezo2020variationalreview,bharti2022noisy,chang2025primer}, as well as digitized adiabatic evolutions~\cite{albash2018adiabatic} or Trotterized quantum dynamics.

Crucially, we note that $f(\rho)$ is invariant under permutations acting on the initial state~\cite{meyer2023exploiting, nguyen2022atheory, schatzki2022theoretical}, as 
\begin{align}
     f(R(\sigma) \rho R\ad(\sigma))&=\Tr[U  R(\sigma) \rho R\ad(\sigma)  U\ad  O ] 
    \nonumber\\ & = \Tr[R(\sigma)  U  \rho  U\ad O R\ad(\sigma) ] 
    \nonumber \\ & = \Tr[U \rho U\ad O]=f(\rho)\,,
\end{align}
where we use the fact that $U$ and $O$ commute with $R(\sigma)$. Hence,  the final expectation value remains unchanged under the action of any group element $\sigma \in S_n$ on the input state $\rho$.

\begin{figure*}[t]
    \subfloat[$U$]{\includegraphics[width = 0.33\textwidth]{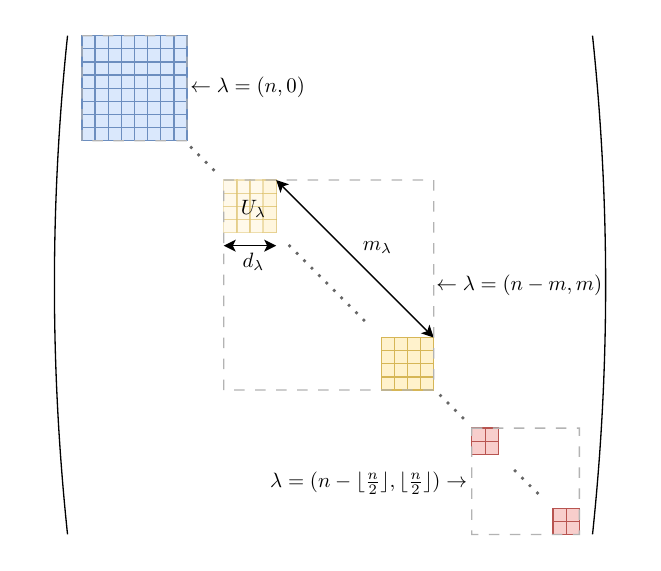}}
    \subfloat[$R(\sg)$]{\includegraphics[width = 0.33\textwidth]{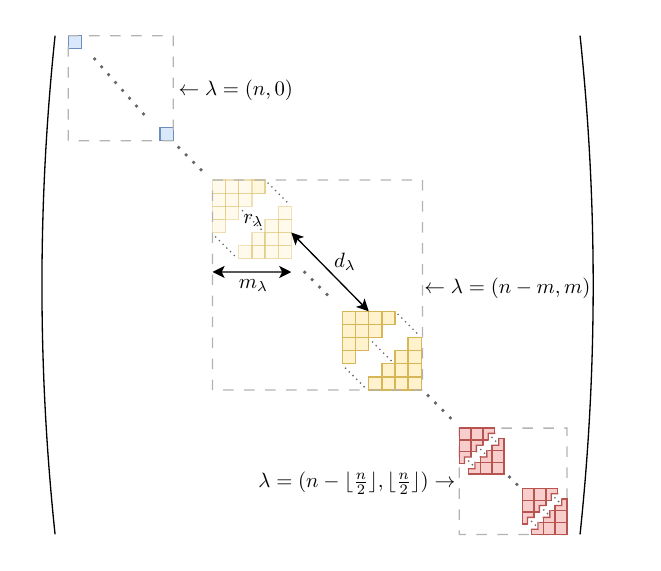}}
    \subfloat[\label{fig:rho_symmetric}$\rho$]{\includegraphics[width = 0.33\textwidth]{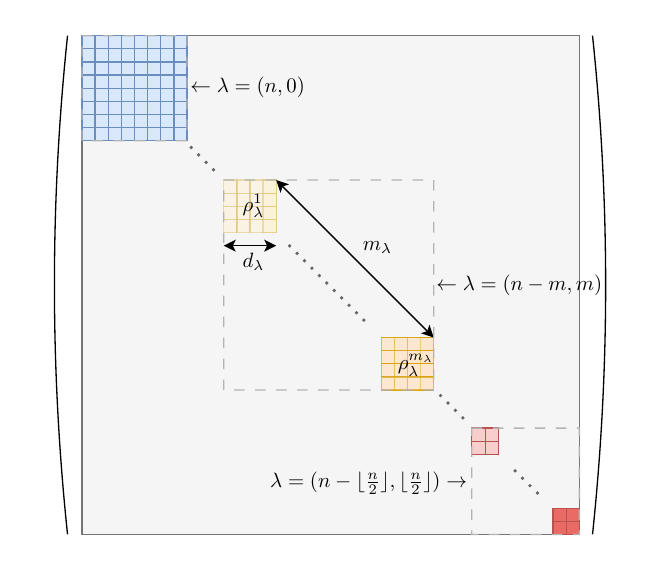}}
\caption{\textbf{Schematic representation of matrices in the basis that block-diagonalizes the group's action.} Each of the dashed blocks corresponds to the irreps labeled by an index $\lm = (n-m, m)$ where $m = 0,\dots, \lfloor \frac{n}{2}\rfloor$. The irreps $\lm$ of $U$ are of the dimension $d_\lm$ and multiplicity $m_\lm$, while the irrep of $R(\sg)$ have a dimension $m_\lm$ and multiplicity $d_\lm$. Note that as illustrated by the gray region in (c), the initial states $\rho$ need not be block-diagonal as it is not necessarily $S_n$-equivariant.\label{fig:rep_theory}}
\end{figure*} 

Our general strategy for simulating $f(\rho)$ in Eq.~\eqref{eq:expec} is to Heisenberg evolve the measurement operator through the layers of the circuit. Crucially, we will leverage the fact that due to their $S_n$-equivariance, these operators have an underlying polynomially-sized structure  (as expected from Eq.~\eqref{eq:dim-com}) that can be understood and exploited with tools from representation theory. 

In particular, we recall Maschke's theorem~\cite{fulton1991representation}, which states that there exists a change of basis, known as the Schur transform, under which the finite-dimensional representation of a group can be expressed as a direct sum of irreducible representations (irreps) 
\begin{equation}\label{eq:irrep-decomp}
    R(\sg\!\in \!S_n) \cong \bigoplus_{\lm}\!\!\!\!\bigoplus_{\nu=1}^{{\rm dim}(r_\lm)}\!\!\!r_\lm(\sg) = \bigoplus_\lm r_\lm(\sg) \otimes \eye_{{\rm mult}(r_\lm)},
\end{equation}
where $\lm$ corresponds to the irreps of $S_n$, $r_\lm$ the irreps of dimension ${\rm dim}(r_\lm)$ and ${\rm mult}(r_\lm)$ its multiplicity. One can prove that~\cite{fulton1991representation}
\begin{equation}
    {\rm dim}(r_\lm) = \frac{n!(n-2m + 1)!}{(n - m + 1)! m! (n-2m)! }\,,
    \label{eq:m_lm}
\end{equation}
and
\begin{equation}
    {\rm mult}(r_\lm) = n - 2m + 1\in\OC(n)\;, 
    \label{eq:d_lm}
\end{equation}
from which we recover $\sum_{\lm} {\rm dim}(r_\lm)\cdot {\rm mult}(r_\lm) = 2^n$. Importantly, the irreps $\lm$ of $S_n$ appearing in Eq.~\eqref{eq:irrep-decomp} can be visualized with a two-row Young diagram~\cite{fulton1997young} and parameterized with a non-negative integer $m$, as $\lm \equiv \lm(m) = (n-m, m)$ where $m = 0,\dots, \lfloor \frac{n}{2} \rfloor$\footnote{Note that the irreps $\lm$ are only labeled by two values $\lm = (\lm_1, \lm_2)$ as we focus on qubits. However, if the quantum system is composed of $D$-dimensional qudits, each irrep is labeled by $D$ values, $\lm_1,\dots, \lm _D$ and the Young diagram will have $D$ rows.}. 

In the same basis under which the representation is block-diagonal, we can also express the $S_n$-equivariant unitaries $U$ and measurement operator $O$ as
\begin{equation}
    U  \cong \bigoplus_{\lm} \eye_{m_\lm} \otimes U_\lm \;,
    \label{eq:U_block}
\end{equation}
and
\begin{equation}
    O \cong \bigoplus_{\lm} \eye_{m_\lm} \otimes O_\lm\;, 
    \label{eq:O_block}
\end{equation}
where $U_\lm\ $ and $O_\lm$ are irrep blocks labeled by $\lm$ with size $d_\lm \times d_\lm$ and multiplicity $m_\lm$. Since Eq.~\eqref{eq:O_block} holds for all Hermitian $S_n$ equivariant generators of the circuit in Eq.~\eqref{eq:U_expression}, then we find that each irrep block of the circuit $U_\lm$ can be computed independently as 
\begin{equation}
    U_\lm  = \prod_{\ell = 1}^L e^{-i (H_\ell)_\lm}\,, 
    \label{eq:U_lambda}
\end{equation}
where we expressed $H_\ell \cong \bigoplus_\lm \eye_{m_\lm} \otimes (H_\ell)_\lm$ with $(H_\ell)_\lm$ being a Hermitian matrix of size $d_\lm \times d_\lm$. Then, we note that the block-diagonal forms of Eq.~\eqref{eq:irrep-decomp} and Eqs.~\eqref{eq:U_block}--\eqref{eq:O_block} make it evident that the commutation relations of Eq.~\eqref{eq:comm} must hold. In addition, this also means that the matrices $U_\lambda$ acts on spaces of dimension $d_\lm = {\rm mult}(r_\lm)$, while the multiplicity of each operator block is $m_\lm = {\rm dim}(r_\lm)$. Thus, relative to the representation-theoretic convention of Eq.~\eqref{eq:irrep-decomp} for $S_n$, the roles of dimension and multiplicity are interchanged for $S_n$-equivariant operators.  In Fig.~\ref{fig:rep_theory}, we illustrate the decomposition of $R(\sg)$ and $U$.

Finally, we note that under the action of the $S_n$-equivariant quantum circuit, the Hilbert space decomposes as 
\begin{equation}
    \HC \cong \bigoplus_{\lm} \Cbb^{m_\lm} \otimes \HC_\lm = \bigoplus_\lm \bigoplus_{p_\lm= 1}^{m_\lm} \HC_\lm^{p_\lm}\;,
\end{equation}
where each subspace $\HC^{p_\lm}_\lm$ of dimension $d_\lm$  is spanned by the so-called Schur basis $\ket{\lm, q_\lm, p_\lm}$, i.e.,
\begin{equation}
    \HC^{p_\lm}_\lm = \spn\left\{\;\ket{\lm, p_\lm, q_\lm}\;\right\}^{d_\lm - 1}_{q_\lm = 0}\;. 
\end{equation}
The  states $\ket{\lm, q_\lm, p_\lm}$ are labeled with three indices, the irrep label $\lm$, the multiplicity label $p_\lm = 0,\dots, m_\lm-1$, and the dimension label  $q_\lm = 0,\dots,d_\lm-1$. In the next section, we provide additional useful details for the Schur basis.

\begin{table*}[t]
\renewcommand{\arraystretch}{1.7}

    \centering
    \begin{tabular}{c|c|c|c}
             Partition &  Young Tableau ($p_\lm = p_\lm^0$) & Dimension label, $q_\lm$  & Schur basis \\
         \hline 
        \multirow{5}{*}{$\lm = (4, 0)$}   & \multirow{ 5}{*}{ \begin{ytableau} 1 & 2 & 3 & 4 \end{ytableau} } & 0 &  {\color{blue}$\ket{0000}$} \\ 
        & & 1 & {\color{blue}$\frac{1}{2}\left(\ket{0001} + \ket{0010} + \ket{0100} + \ket{1000} \right) $}\\ 
        & & 2 & {\color{blue}$\frac{1}{\sqrt{6}}\left(\ket{0011} + \ket{0110} + \ket{1100} + \ket{1001} + \ket{0101} + \ket{1010} \right) $}\\ 
        & & 3 & {\color{blue}$\frac{1}{2}\left(\ket{0111} + \ket{1011} + \ket{1101} + \ket{1110} \right) $}\\ 
        & & 4 & {\color{blue}$\ket{1111}$}\\ 
        \hline
         \multirow{3}{*}{$\lm = (3, 1)$} &   \multirow{3}{*}{\begin{ytableau} 1 & 3 & 4 \\ 2  \end{ytableau} } & 0 & {\color{red}$\frac{1}{\sqrt{2}}(\ket{01} - \ket{10})$}{\color{blue}$\ket{00}$} \\
        & & 1 & {\color{red}$\frac{1}{2}(\ket{01} - \ket{10})$}{\color{blue}$(\ket{01} + \ket{10})$} \\    
        & & 2& {\color{red}$\frac{1}{\sqrt{2}}(\ket{01} - \ket{10})$}{\color{blue}$\ket{11}$ }  \\[3pt]   
        \hline 
        \vspace{5pt} $\lm = (2, 2)$ & \vspace{1cm} \begin{ytableau} 1 & 3 \\ 2 & 4 \end{ytableau} & 0 & {\color{red}$\frac{1}{2}\left(\ket{01} - \ket{10}\right)\left(\ket{01} - \ket{10}\right)$}\\  
    \end{tabular}
    \caption{\textbf{Partitions of 4 qubits enumerated by Young tableaux and their corresponding Schur basis.} For each partition $\lm = (n-m, m)$, a representative canonical Young tableau ($p_\lm = p_\lm^0$) is shown with the corresponding Schur basis state, expressed in the computational basis. As shown above, the Schur basis consists of an antisymmetric state (in red) on $2m$ qubits and a symmetric state (in blue) on $n - 2m$ qubits defined as Eq.~\eqref{eq:schur_basis_expression}. The dimension label $q_\lm$ also characterizes the Hamming weight in the symmetric part of the basis.  }
    \label{tab:Schur_basis}
\end{table*}

\subsection{Schur basis\label{sec:schur_basis}}
As previously mentioned, the Schur basis is related to the computational basis through the Schur transform, denoted as $\QST$. Accordingly, we can write~\cite{bacon2005quantum}:
\begin{equation}
 \ket{\lm, p_\lm, q_\lm} =  \sum_{i_1,\dots,i_n}[\QST]^{\lm, p_\lm, q_\lm}_{i_1,\dots, i_n} \ket{i_1\dots i_n}
\end{equation}
where $\ket{i_1, \dots, i_n}$ is a computational basis state labeled by a bitstring $i_1  \cdots i_n \in \{0, 1\}^{n}$.  The entries of $\QST$ can be constructed using Young diagrams and Young tableaux.

A Young tableau is defined as a filling of the Young diagram. In particular, it is called standard if its entries increase strictly each row and strictly down each column with integers in the set $\{1,..,n\}$. For each standard Young tableau, we can associate a multiplicity $p_\lm$ ~\cite{fulton1997young}. Then, given a Young tableau $p_\lm$, we define ${\rm Row}(p_\lm)$ and ${\rm Col}(p_\lm)$ as the subgroup of permutations obtained by permuting the integers within each row and column of $p_\lm$, respectively. In Table~\ref{tab:Schur_basis}, we present an example of Young tableaux for 4 qubits, corresponding to three distinct partitions. Each partition defines a Young tableau and labels a distinct irrep, $\lm$. 

Throughout this work, among the different standard Young tableaux,  we refer to the Young tableau presented in Table~\ref{tab:Schur_basis} as the canonical tableau, denoted by the multiplicity label $p_\lm^0$. This choice follows the convention of Appendix A in Ref.~\cite{anschuetz2022efficient}, where the corresponding Young tableau is filled in a fixed order---first column-wise, then row-wise. Although we show only one specific Young tableau corresponding to a particular multiplicity label $p_\lm$, other Young tableaux corresponding to the same irrep $\lm$ can be constructed by permuting the qubit indices within the boxes, yielding different multiplicity labels $p_\lm$.

Using the qubit-defining representation $R$ of the permutation action $\sg \in S_n$, introduced in Eq.~\eqref{eq:sn_representation}, and given a Young tableau $p_\lm$, we define the Young symmetrizer~\cite{bacon2005quantum} as the operator  
\begin{align}
\mathbf{S}_{p_\lm} \! & = \frac{m_\lm}{n!} a(p_\lm) b(p_\lm) 
\label{eq:young_symmetrizer}
\end{align}
where $a(p_\lm)$ symmetrizes all the components within each row of the tableau and $b(p_\lm)$ antisymmetrizes all the components within each column~\cite{howe2022irreducible, fulton1997young}. Explicitly, these operators are given by
\begin{align}
    a(p_\lm) & = \left(\sum_{\sg_r\in {\rm Row}(p_\lm)}\mkern-6mu R(\sg_r)\right), \\ 
    b(p_\lm) & =  \Bigg(\sum_{\sg_c\in {\rm Col}(p_\lm)}   {\rm sgn}(\sg_c) R(\sigma_c) \Bigg)  
\end{align}
where $ {\rm sgn}(\sg)$ is the parity of the permutation $\sigma$. 

The Young symmetrizer  projects the Hilbert space $\HC$ onto the subspace $\HC^{p_\lm}_{\lm}$, which is spanned by the orthonormal Schur basis $\{ \ket{\lm, p_\lm, q_\lm} \}^{d_\lm - 1}_{q_\lm = 0}$. 
By applying $\mathbf{S}_{p_\lm}$ to the standard computational basis of $\HC$, its image corresponds to the desired copy of the irrep subspace labeled by $\lm$. From these vectors, one can explicitly construct the Schur basis associated with a given Young tableau $p_\lm$. In particular, choosing the canonical Young tableau  yields the basis states: 
\begin{equation}
    \ket{\lm, p^0_\lm, q_\lm } = \ket{\Psi}^{\otimes m} \otimes \ket{\Sigma^{(n-2m)}_{q_\lm}}\,,    \label{eq:schur_basis_expression}
\end{equation}
where  $\ket{\Psi}$ is the two-qubit antisymmetric singlet state given by 
\begin{equation}
    \ket{\Psi} =  \frac{1}{\sqrt{2}}\left(\ket{01} - \ket{10}\right)\,,
    \label{eq:antisymmetric}
\end{equation}
and $\ket{\Sigma^{(n-2m)}_{q_\lm}}$ is the Dicke state on $n - 2m$ qubits with Hamming weight $q_\lm$ for $q_\lm = 0,\dots,d_\lm -1$~\cite{bartschi2019deterministic}. 
More precisely, the Dicke state $\ket{\Sigma^{(n-2m)}_{q_\lm}}$ is defined as the normalized, equal superposition of all computational basis states on $(n-2m)$ qubits with Hamming weight $q_\lm$, representing the fully symmetric part
\begin{equation}
    \ket{\Sigma_{q_\lm}^{(n-2m)}} = \binom{n-2m}{q_\lm}^{-\frac{1}{2}} \sum_{\substack{x \in \{0, 1\}^{n-2m} \\ {\rm HW}(x) = q_\lm}}\ket{x}\,.
\end{equation}
Here, $\ket{x}$ denotes a computational basis state and ${\rm HW}(x)$ its Hamming weight.

Thus, for the Young tableau with label $\lm = (n-m, m)$, the canonical Schur basis consists of $m$ antisymmetric singlet pairs (occupying $2m$ qubits) arising from the column antisymmetrization $b(p_\lm)$, while the remaining $n-2m$ qubits span a fully symmetric subspace generated by the row symmetrization $a(p_\lm)$. An explicit example of such a basis constructed via the Young symmetrizer is shown in Table~\ref{tab:Schur_basis}. 

The total number of Schur basis states to be constructed scales as $\sum_\lm d_\lm \approx (n+2)^2/4$. As the system size $n$ grows, explicitly constructing an orthogonal Schur basis set across different multiplicity sectors becomes increasingly cumbersome. However, since Eqs.~\eqref{eq:U_block} and~\eqref{eq:O_block} show that the same irrep blocks are repeated across the multiplicity indices, it is sufficient to work with the canonical Schur basis when constructing $U_\lm$ and $O_\lm$.

\section{Classical Simulation of $S_n$-equivariant unitaries\label{sec:classical_simulation}}

In this section, we will present an efficient algorithm that can be used to estimate  $f(\rho)$. To begin, let us note that Eqs.~\eqref{eq:U_block} and~\eqref{eq:O_block} imply that this function can be expressed as a sum of the trace terms on the individual irrep blocks. That is, 
\begin{align}
    f(\rho)
    & = \sum_{\lm}\sum_{p_\lm = 1}^{m_\lm}\Tr[\rho_{\lm}^{p_\lm} U_\lm\ad O_\lm U_\lm] \nonumber \\ 
    & = \sum_{\lm}\Tr[\left(\sum_{p_\lm = 1}^{m_\lm} \rho^{p_\lm}_\lm\right)\widetilde{O}_\lm]\;,   
    \label{eq:loss_block}
\end{align}
with $\widetilde{O}_\lm = U_\lm \ad O_\lm U_\lm$ the Heisenberg evolved operator within each irrep. In this equation, $\rho_\lm^{p_\lm}$ denotes the projection of $\rho$ onto the subspace $\HC_{\lm}^{p_\lm}$. Note that the initial state $\rho$ is not necessarily $S_n$-equivariant, and therefore cannot, in general, be block-diagonalized in the Schur basis, as illustrated in Fig.~\ref{fig:rho_symmetric}. This is precisely why we work in the Heisenberg picture of backward-in-time evolution of measurement operators.

In what follows, we divide the task of estimating $f(\rho)$ into two parts. First, we will  focus on the evaluation of $\widetilde{O}_\lm$, which will require us to compute the projection of $O$ and each gate generator $H_\ell$ into the irreps and Heisenberg evolve the measurement operator. Second, we will discuss  how to obtain $\rho_\lm^{p_\lm}$.

For the first task, one needs to block-diagonalize the circuit generator and the measurement operators in the Schur basis. 
As a solution to this problem, Ref.~\cite{anschuetz2022efficient} proposed an algorithm to compute the matrix elements of any arbitrary $S_n$-equivariant operator using tensor network diagram contraction, with computational complexity  of $\order{n^7}$. Although this complexity scales polynomially with respect to the system size $n$, it can rapidly become prohibitively expensive to deploy at even modest system sizes. Our approach is to instead focus on block-diagonalizing in the basis obtained from the canonical Young tableau for specific $S_n$-equivariant operators that are physically motivated, rather than trying to work with arbitrary ones (see Refs.~\cite{sauvage2024classical,schatzki2022theoretical} for similar approaches).   

Our decision is motivated by the fact that, in practice, most of the $S_n$-equivariant generators in $\GC_{S_n}$ are at most two-local~\cite{schatzki2022theoretical, cervia2021lipkin, li2024enforcing, zheng2025toward}. For instance, if we are implementing a Trotterized evolution, the terms in a physical Hamiltonian are one- or two-bodied. Hence, one possible choice for the set of local $S_n$-equivariant generators $\GC$ is 
\begin{equation}
    \GC_{S_n} = \left\{\frac{1}{n}\sum_{j=1}^n P_j\;, \frac{2}{n(n-1)}\sum_{1\le k < j \le n} P_jP_k\right\}\,,
\label{eq:Sn_generators}
\end{equation}
where $P_j \in \{X, Y, Z\}$ is a fixed Pauli matrix at the $j$-th qubit.
Similarly, we can also consider a set of  equivariant observables $\chi$, given by
\begin{equation}\chi = \GC_{S_n}\cup \left\{\bigotimes_{j=1}^nP_j\;\right\}\,,
\label{eq:Sn_observables}
\end{equation}
which also contains the global measurement operators.

Here, we also find it important to note that the choice of generators $\GC_{S_n}$, while being local, leads to subspace-universal circuits within each invariant subspace. That is, unitaries as in Eq.~\eqref{eq:U_expression} with generators taken from $\GC_{S_n}$ will be subspace controllable and can generate any arbitrary unitaries in each subspace $\lm$ (see Theorem 1 in Ref.~\cite{kazi2023universality}).

\subsection{$S_n$-equivariant operators in Schur basis \label{sec:Sn_operators}}
Let $A$ be an $S_n$-equivariant operator such as those in Eqs.~\eqref{eq:Sn_generators} and \eqref{eq:Sn_observables}. As previously mentioned, in the Schur basis, $A$ admits the block-diagonal form $A \cong \bigoplus_\lm \eye_{m_\lm} \otimes A_\lm$, where the matrix entries at position $(q_\lm, q'_\lm)$, denoted as  $\left(A_\lm\right)_{q_\lm, q'_\lm}$, can be found as the inner product between $A$ and two Schur basis states within a fixed multiplicity. In particular, we will focus on the case when $p_\lm = p_\lm^0$, so that
\begin{equation}
\left(A_\lm\right)_{q_\lm, q'_\lm} = \bramatket{\lm, p^0_\lm, q_\lm}{A}{\lm, p^0_\lm, q'_\lm}. 
\end{equation}
Thus, we can explicitly compute these matrix elements by considering the action of $A$ on the symmetric and antisymmetric parts of the canonical tableau. 

We first consider $A$ to be solely composed of Pauli-$Z$ operators. 
As the canonical Schur basis $\ket{\lm, p^0_\lm, q_\lm}$ is characterized by the Hamming weight of the computational basis in the symmetric part, given by $h_{\lm,q_\lm} = q_\lm + m$, Pauli-$Z$ operators act diagonally in the Schur basis. Using the eigenvalues of $A$ expressed as a function of the Hamming weight in the computational basis (c.f. Lemma 4 in Ref.~\cite{kazi2023universality}), the matrix elements in each irrep  take the form 
\begin{align}
\left(A_\lm\right)_{q_\lm q'_\lm} =         \frac{n^2 - n - 4nh_{\lm,q_\lm}  + 4h_{\lm,q_\lm}^2}{n(n-1)}\dl_{q_\lm q'_\lm}\,,
\end{align}
for $A = \frac{2}{n(n-1)}\sum_{j < k} Z_j Z_k$;
\begin{align}
\left(A_\lm\right)_{q_\lm q'_\lm}    = 
    \
        \left(1 - \frac{2 h_{\lm,q_\lm}}{n}\right) \dl_{q_\lm q'_\lm}\,,
\end{align}
for $A = \frac{1}{n}\sum_{i} Z_i$; and 
\begin{align}
\left(A_\lm\right)_{q_\lm q'_\lm}    = 
    (-1)^{h_{\lm,q_\lm}} \dl_{q_\lm q'_\lm} \,,
\end{align} 
for $A = \bigotimes_{i=1}^n Z_i$.

Next, we analyze the operators composed of Pauli-$X$ operators, beginning with the operator $A = X^{\otimes n}$. The action of $A$ simply flips each of the computational basis bits, while changing the sign of the antisymmetric state. As a result, we find
\begin{align}
    A\Big(\ket{\Psi}^{\otimes m} & \otimes \ket{\Sigma^{(n-2m)}_{q_\lm}} \Big) \nonumber \\ & = (-1)^m \ket{\Psi}^{\otimes m } \otimes \ket{\Sigma^{(n-2m)}_{n-2m -q_\lm}}\;.
\end{align}
Therefore, each block $\lm$ has an anti-diagonal structure, with entries consisting solely of $1$ or $-1$: 
\begin{equation}
    \left(A_\lm\right)_{q_\lm q'_\lm} = (-1)^{m}\dl_{q_\lm, n - 2m - q'_\lm}\,.
\end{equation}
Then, taking $A = \frac{1}{n}\sum_{i=1}^n X_i$, the antisymmetric part $\ket{\Psi}$ always vanishes, and the action on the symmetric part shifts the Hamming weight by $\pm 1$, yielding
\begin{align}
\left(A_\lm\right)_{q_\lm q'_\lm}    \!= \!\frac{\al^{-}(\lm, q_\lm)}{n}\dl_{q_\lm -1, q'_\lm}  \! +\! \frac{\al^{+}(\lm, q_\lm)}{n} \dl_{q_\lm + 1, q'_\lm}\,,  
\end{align}
where we define the coefficients 
\begin{align}
    \al^{-}(\lm, q_\lm) & = \sqrt{q_\lm (n - 2m - (q_\lm - 1))}\,, \\ 
    \al^{+}(\lm, q_\lm) & = \sqrt{(q_\lm + 1) (n - 2m - q_\lm)}\,. 
\end{align} 
In this case, each block $\lm$ will be a tridiagonal matrix with the main diagonal hollow and $\order{2n}$ non-zero off-diagonal elements.  

The previous examples illustrate a general pattern: Schur blocks of at most two-local $S_n$-equivariant operators, as well as the global $\bigotimes_{i=1}^n P_i$ ones, are sparse---being diagonal, anti-diagonal or banded---with at most $\order{n}$ entries per block. As a result, all blocks $A_\lm$ can be constructed with total classical cost $\order{n^2}$. 
The representations for all other operators of interest, as well as additional details, are provided in Appendix~\ref{adx:explicit_blocks}.

\subsection{Classical simulation of $S_n$-equivariant circuits}

Using the results from the previous section, we now have a practical way to compute the Heisenberg-evolved measurement operators $\widetilde{O}_\lm$ on each of the size $d_\lm \times d_\lm$ blocks, for all $\left\lfloor \frac{n}{2}\right\rfloor +1 $ distinct irrep labels $\lm$.

Specifically, we begin by projecting the observable $O$ and each circuit generator onto each irrep blocks using the techniques previously described. Then, let us focus on a single irrep. If we compute the single layer exponential $e^{-i (H_\ell)_\lm}$ and use matrix product in the Heisenberg evolution, the naive cost would scale as $\order{L  d_\lm^3}$ with dominating matrix exponentiation cost. For two-local symmetrized Pauli strings, the irreducible blocks $(H_\ell)_\lambda$ are, in the worst case, banded Hermitian matrices with constant bandwidth. Hence, they can be diagonalized in $\order{d_\lm^2}$ time, allowing us to express them as
\begin{equation}\label{eq:diag}
    (H_\ell)_\lm = Q_{\ell, \lm} \Lm_{\ell, \lm} (Q_{\ell,\lm})\ad,
\end{equation} with $Q_{\ell, \lm}$ unitary and $\Lm_{\lm, \ell}$ diagonal. Then, the unitary in the irrep is obtained as
\begin{equation}
    e^{-i (H_\ell)_\lm}
      = Q_{\ell,\lm}\,e^{-i \Lm_{\ell,\lm} } \left(Q_{\ell, \lm} \right)\ad\,,
\end{equation}
where we recall that the exponentiation of a diagonal matrix has a cost of $\order{d_\lm}$. Thus, given access to Eq.~\eqref{eq:diag}, we can compute a single layer Heisenberg evolution in time $\order{d_\lm^\omega}$ with  $\omega$  the matrix multiplication exponent, taking a value between 2.37 and 3 depending on the method used. Summing over irreps, layers, and taking into account the cost of diagonalizing all generators leads to the following theorem.

\begin{theorem}[Complexity of Heisenberg evolution\label{thm:time_complexity}]
    Consider an $S_n$-equivariant circuit $U$ with a set of one- and two-local generators defined as in  Eq.~\eqref{eq:Sn_generators}, and an $S_n$-equivariant observable $O$ from the set Eq.~\eqref{eq:Sn_observables}.  The Heisenberg-evolved operator $U\ad O U$ can be classically evaluated in the Schur basis with time cost scaling as 
    \begin{equation}
        N_{\CT } \in \order{n^3 + L n^{\om + 1}}\;,
        \label{eq:time_complexity}
    \end{equation} 
    and the memory cost scaling as 
    \begin{equation}
        N_{\rm memory} \in \order{n^3}\;. 
    \end{equation}
\end{theorem}

The first term in  Eq.~\eqref{eq:time_complexity} arises from the diagonalization of the one- and two-body generators in Schur basis. For each irrep $\lm$, this requires $\order{d_\lm^2}$ operations, and summing  over contributions across all irreps yields $\sum_\lm d_\lm^2 = \sum_{m = 1}^{
\lfloor n/2\rfloor } (n-2m +1)^2 \in \order{n^3}$. The second term accounts for the matrix multiplication cost to obtain $U_\lm$, as well as $U_\lm\ad O_\lm U_\lm$ across all irreps. For each irrep, this costs $\order{L d_\lm^\omega}$, leading to the total complexity $\sum_\lm L d_\lm^\omega \in \order{L n^{\omega + 1}}$. In addition, the diagonalizing matrices $Q_{\ell, \lm}$ of dimension $d_\lm \times d_\lm$ must be stored for all $\lm$ and $\ell$, resulting in the memory cost of $\order{n^3}$.  Here, we underline that the theorem assumes that the circuit generators are drawn from a fixed set. If, instead, each layer involve a new linear combination of generators from the set, the diagonalization should be performed independently at each layer, resulting in an overall time complexity of $\order{Ln^3 + L n^{\omega + 1}}$ and a memory complexity $\order{Ln^3}$.

A priori, the diagonalization of the circuit generators needs to be performed only once as a preprocessing step for each value of $n$. Hence, assuming access to such diagonalization, the cost of simulation per layer is simply given by $\order{ n^{\om + 1}}$. However, as we will observe in Section~\ref{sec:numerics},  the practical runtime scales more favorably than the theoretical upper bound.

At this point, we note that the results in Theorem~\ref{thm:time_complexity} are closely related to those in Lemmas~3 and 4 in Ref.~\cite{anschuetz2022efficient}, but refined here for circuits generated by at most two-local Pauli generators. In fact, below we show how to extend our theorem for $k$-local Pauli generators with $k \in \order{1}$. 

To finish, we note that if the components $\sum_{p_\lm = 1}^{m_\lm} \rho^{p_\lm}_\lm$ are known, the previous algorithm readily allows us to compute $f(\rho)$ with an additional step of matrix multiplication and tracing. While in general, we might not have access to the decomposition of the input state within each irrep, in some special cases we can efficiently compute them. This is, for instance, the case when $\rho$ is $S_n$-equivariant itself, then $\rho^{p_\lm}_\lm=\rho_\lm$ (i.e., the same component repeats across multiplicities) or when it is a simple state such as  $\rho=\dya{0}^{\otimes n}$. Here,  the following corollary follows.
\begin{corollary}
    Assume that we are given access to the components of the general input state $\rho$ within each irrep, as well as to the diagonalized circuit generators. Then, we can compute $f(\rho)$ for an $L$-layered circuit with $L\in\OC(1)$ in time scaling as 
        \begin{equation}
        N_{\CT } \in \order{ n^{\om + 1}}\;.
\label{eq:time_complexity-f}
    \end{equation} 
\end{corollary}

More generally, we will see below that  if $\rho$ is a generic state that  lives in a quantum computer, we can obtain its irrep projections via classical shadows techniques. This incurs some additional complexity in the form of classical shadows samples, as well as some classical post-processing. Still, even in this case our algorithm remains efficient across the board. We explore this setting in Section~\ref{sec:shadow}, with additional details given in Appendix~\ref{adx:pi_cs}.

\subsection{Generic $k$-local symmetrized Pauli operators\label{sec:k_local_PS}}
Until now, we have restricted our analysis to the case when the elements in $\GC$ are at most two-local, where we can analytically find their irrep projections.  In this section, we extend the framework to symmetrized $k$-local Pauli strings $\TC_{S_n}(P_{\vec{k}})$ with $k \in \order{1}$. In particular, we refer the reader to  Appendix~\ref{adx:k_local_generators} for an explicit algorithm that numerically evaluates the matrix expression of $\TC_{S_n}(P_{\vec{k}})$ in the Schur basis for an arbitrary $\vec{k}$. 

The following theorem states the computational complexity bound for evaluating the matrix representation of the $k$-local symmetrized Pauli strings in the Schur basis. 
\begin{theorem}
Consider a symmetrized Pauli operator $\TC_{S_n}(P_{\vec{k}})$ with $P_{\vec{k}}$ defined in Eq.~\eqref{eq:Pk}. Given that $k \in \order{1}$, the total time complexity to find the matrix representation in the Schur basis scales at most as $\order{n^2}$. 
\end{theorem}
As a result, for constant locality $k$, the computational cost of evaluating the matrix elements of $k$-local symmetrized Pauli strings is negligible compared to the cost of computing $U_\lm$ for all irrep labels $\lm$.  
Moreover, for a $k$-local Pauli string, the corresponding matrix block $H_\lm$ is at most $k$-banded. Therefore, its eigendecomposition can be performed with a computational cost that also scales at most as $\order{n^2}$, and the overall time complexity stated in Theorem~\ref{thm:time_complexity} remains valid. We note that this bound is tighter compared to the bound of $\order{n^4}$ provided in Appendix E of Ref.~\cite{anschuetz2022efficient}. 
In contrast, for generic Pauli strings with $k \in \order{n}$,  we recover the time complexity of $\order{n^4}$.

\subsection{Initial state acquisition via classical shadows~\label{sec:shadow}}

In this section, we discuss how a quantum computer can be used to estimate the irrep projection of $\rho$ via classical shadows~\cite{elben2022randomized}.  As usual, it is fundamental to pick a shadow tomographic procedure that aligns well with the information we are trying to extract.  For $S_n$-equivariant quantum simulations, the initial state can be efficiently acquired using PI-CS, a shadow protocol tailored to permutation-invariant circuits~\cite{sauvage2024classical}. We briefly review the PI-CS protocol and analyze the associated resource requirements, thereby completing the complexity analysis for $S_n$-equivariant simulations over generic states. Additional details on PI-CS are provided in Appendix~\ref{adx:pi_cs}.

We start by considering an arbitrary initial state $\rho$, which is not necessarily $S_n$-equivariant. The last equality in Eq.~\eqref{eq:loss_block} implies that it is sufficient to retrieve the sum of the initial states over multiplicities for each irrep block $\lm$\footnote{The fact that we do not require each individual projection onto the multiplicity blocks is crucial, as that would require us to store an exponential amount of information.}. A variant of PI-CS, dubbed Deep PI-CS, performs this task by block-diagonalizing the initial state with the quantum Schur transform (QST) circuit $\QST$ of depth scaling linearly with respect to system size $n$ and polylogarithmically in the inverse precision~\cite{bacon2005quantum, bacon2006efficient, kirby2017practical, krovi2019efficient}. This leads to the following theorem for the quantum sample complexity:
\begin{lemma}[Quantum circuit and sample complexity for generic states]

 Given a generic state $\rho$, the expectation value of a set of $M$ distinct $S_n$-equivariant observables $\{O_i\}_{i=1}^M$ can be estimated up to additive accuracy $\ep$, and with a success probability $1 -\dl$, via deep PI-CS using a number of quantum samples scaling as 
    \begin{equation}
        N_{\QS} \in \order{\log(\frac{M}{\dl})\frac{n^2}{\ep^2} \max_i{\norm{O_i}^2_{\infty}}}\;, 
    \end{equation}
where $\norm{\cdot}_{\infty}$ denotes the spectral norm of the matrix. The procedure requires a quantum Schur transform circuit of depth $\order{n \poly\left(\log(\epsilon^{-1}_{\rm QST})\right)}$  where $\ep_{\rm QST}$ denotes the implementation error of the Schur transform.
\label{thm:classical_shadow_PI}
\end{lemma}

If the initial state is $S_n$-equivariant, it can also be decomposed as $\rho \cong \bigoplus_{\lm} \eye_{m_\lm} \otimes \rho_\lm$ with $\rho_\lm$ a $d_\lm \times d_\lm$ matrix as well. In this specific case, we can use shallow PI-CS, or symmetrized PI-CS. Here, there is no need to apply a QST, reducing the quantum circuit complexity. We find that the following lemma holds.
\begin{lemma}[Quantum circuit and sample complexity for $S_n$-equivariant states]
Given an $S_n$-equivariant state, the expectation value of a set of $M$ distinct $S_n$-equivariant observables $\{O_i\}_{i=1}^M$ can be estimated up to additive accuracy $\ep$, and with a success probability $1 -\dl$, via symmetrized PI-CS using a quantum circuit of constant depth, i.e., $N_{\QT}\in \order{1}$, and a number of quantum samples scaling as
    \begin{equation}
        N_{\QS} \in \order{\log\left(\frac{M}{\delta}\right)\frac{n}{\ep^2} \max_i{\norm{O_i}_F^2} }
    \end{equation}
where $\norm{\cdot}_F$ denotes the Frobenius norm of the matrix. 
\end{lemma}

At this point, it is important to note that unlike deep PI-CS, in which additional quantum circuit and qubits are required, the symmetrized PI-CS requires no additional quantum resources. However, while the shadow post-processing is negligible in deep PI-CS, the symmetrized PI-CS does incur non-trivial classical overhead as one must invert the measurement channel and compute the relevant Clebsch–-Gordan coefficients. This adds an extra classical runtime cost of at least $\order{n^6}$, leading to a trade-off between quantum and classical resource consumption. More details on the PI-CS can be found in Appendix~\ref{adx:pi_cs} and \ref{adx:symmetrized_pi_cs}.

\section{Numerical simulation\label{sec:numerics}} 
To benchmark the classical simulation algorithm introduced in Section~\ref{sec:classical_simulation}, we consider the task of ground state preparation via a digitized adiabatic quantum computing (AQC) protocol~\cite{aharonov2008adiabatic, albash2018adiabatic, barends2016digitized}. The AQC approach relies on the adiabatic theorem~\cite{born1928beweis}, which states that a system initialized in the ground state of the initial Hamiltonian $H_0$ remains in its instantaneous ground state provided the Hamiltonian governing the system varies sufficiently slowly compared to the inverse of the minimum energy gap $\Delta$ between the ground state and the first excited state throughout the full evolution.  

Given a target Hamiltonian $H_1$, the system is initially prepared in the ground state of a simple initial Hamiltonian $H_0$, and then driven adiabatically towards the ground state of $H_1$. 
To this end, we introduce a schedule $s(t): [0, T] \to [0, 1] $, satisfying the boundary condition $s(0) = 0$ and $s(T) = 1$, where $T$ denotes the total annealing time. The system dynamics is governed by a time-dependent Hamiltonian constructed as a linear interpolation between the initial and target Hamiltonians as
\begin{equation}
    H(t) = \left(1 - s(t)\right)H_0 + s(t) H_1\;. 
\end{equation}

Ideally, the evolution operator can be written as: 
\begin{equation}
    U(T)   =  \mathcal{T} \exp\left( - i \int_0^T H(t) dt  \right)\;, 
\end{equation}
where $\mathcal{T}$ denotes the time-ordered operator.
While AQC is naturally implemented on analog quantum simulators, an equivalent procedure can be realized on the digital quantum hardware by discretizing the continuous time evolution. 
In the digitized implementation, the continuous time evolution is approximated by a sequence of discrete time steps: 
\begin{equation}
    U(T) \approx \prod_{j = 1}^{L} e^{ -i H(t_j)\Delta t}\,,
\end{equation}
where $L$  is the number of time steps, $\Delta t = \frac{T}{L}$ the size of the time step, and $t_j = j \Delta t$ discrete times. Each short-time exponential can be approximated by first- or second-order Trotter-Suzuki decomposition.
The final time evolution operator approximately follows the adiabatic trajectory up to controllable Trotter or non-adiabatic errors.

Throughout this work, we choose the initial $S_n$-equivariant Hamiltonian to be $H_0 = -\sum_i^n X_i$, which is a standard and widely adopted choice in AQC protocols.  The system is initialized in the ground state of $H_0$, 
$\rho_0 = \ketbraq{\psi_0}$,  where $\ket{\psi_0} = \ket{+}^{\otimes n } = \frac{1}{\sqrt{2^n}} \sum_{j=0}^{2^n-1} \ket{j}$, which corresponds to an equiprobable superposition over all computational basis states. Then, we take our target $S_n$-equivariant Hamiltonian to be that of the Lipkin--Meshkov--Glick (LMG) model~\cite{lipkin1965validity}, a paradigmatic collective-spin model with infinite-range interactions. The Hamiltonian reads 
\begin{equation}
H_{\rm LMG}
= -\frac{J}{n} \sum_{i<j} \left(X_i X_j + \gm Y_i Y_j\right) + h_z\sum_{i} Z_i\; 
\end{equation}
where $J$ sets the interaction strength, $\gamma$ controls anisotropy between $X$ and $Y$ couplings, and $h_z$ is a longitudinal field along the $z$ direction. 

Since  the initial state $\rho_0$ lies in the fully symmetric irrep subspace labeled by $\lambda = (n,0)$, and both $H_0$ and $H_{\mathrm{LMG}}$ are $S_n$-equivariant, the dynamics remain confined to this symmetric subspace throughout the adiabatic evolution. This restriction is not merely an artificial simplification, but rather arises naturally in several physically relevant settings~\cite{huang2021dynamic, zhang2024heisenberg}. The LMG model within the symmetric sector has been extensively studied in the thermodynamic limit, and its ground-state properties are well understood using semiclassical approaches~\cite{dusuel2004finite, dusuel2005continuous, vidal2006concurrence}.

In order to verify the fidelity of our classical simulation, we numerically compute several observables whose behavior has been characterized analytically in the aforementioned references. All simulations were performed using the Julia~\cite{bezanson2017julia} on a 12-core CPU system with 36~GB of memory. The number of adiabatic time steps is chosen to scale linearly with respect to the system size, $L\in \order{n}$, to ensure the convergence of the digitized adiabatic evolution~\cite{childs2021theory}. Although this scaling is adapted empirically without computing a rigorous bound based on the minimum spectral gap, the numerical results presented validate this choice. 

\begin{figure}[h]
    \centering
    \includegraphics[width=0.97\linewidth]{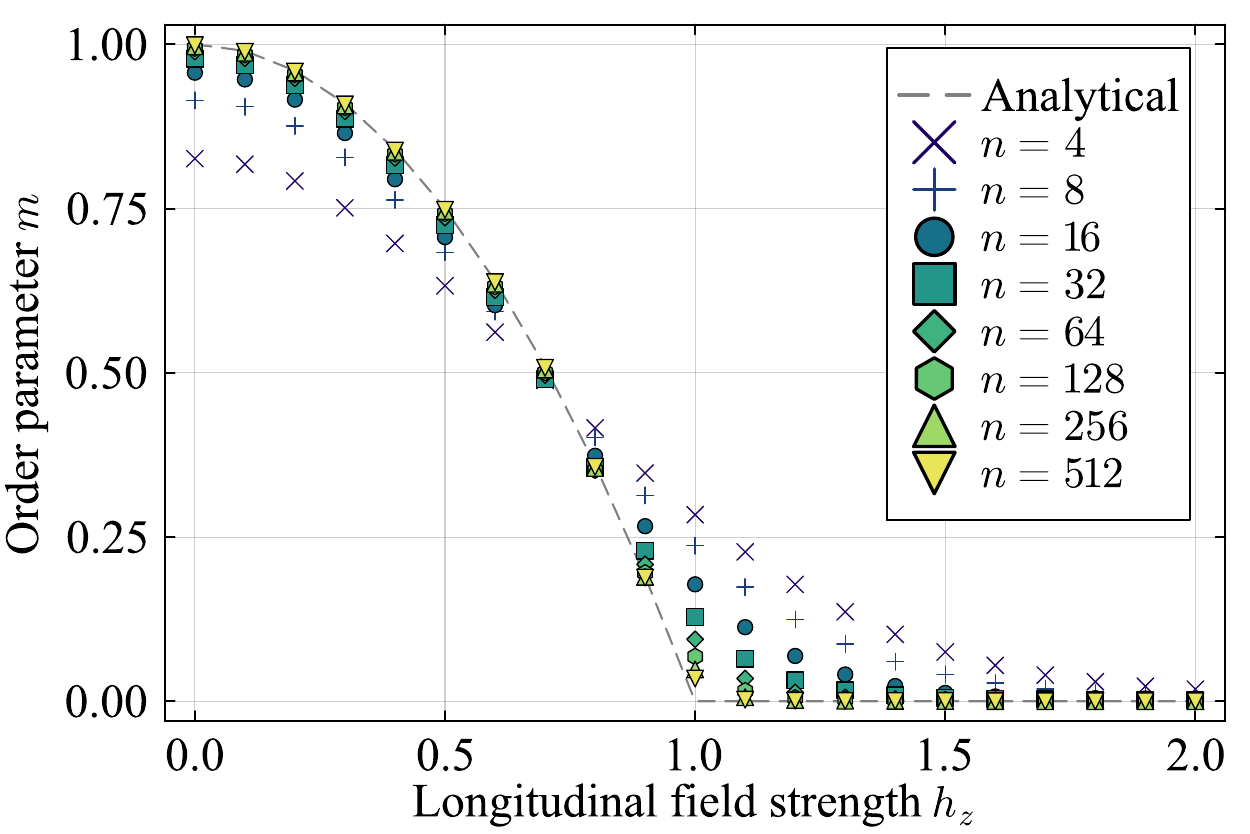}
    \caption{\textbf{Order parameter $m$ as a function of the longitudinal field, $h_z$.} The dashed line corresponds to the analytical expression in the thermodynamic limit, given by Eq.~\eqref{eq:order_params_th}. The result  is shown for $J = 1$ and $\gamma = 0.5$, however, the order parameter is independent of $J$ and $\gamma$, and thus identical for all anisotropy values.  }
    \label{fig:order_params}
\end{figure}

In the anisotropic regime $0\le\gamma\le1$, the LMG model exhibits a second-order quantum phase transition at $h_z=1$. A suitable order parameter is defined in terms of collective magnetization along the $z$-direction, which, in the thermodynamic limit, admits the analytical form
\begin{equation}
    m = 1- \frac{4}{n^2}\avg{(\sum_{i} Z_i)^2} = \begin{cases}
        1 - h_z^2 & \text{for }0 \le h_z < 1\;,\\ 
        0 & \text{for } h_z \ge 1\;, 
    \end{cases}
\label{eq:order_params_th}
\end{equation}
which depends only on the field strength $h_z$. 
In Fig.~\ref{fig:order_params}, we display the order parameter value $m$ computed as a function of $h_z$ for different numbers of qubits and its analytical value in the thermodynamic limit, showing that our algorithms can accurately simulate the LMG model. 

\begin{figure}[h]
    \centering
    
    \subfloat[$\gamma = 0.5$]{\includegraphics[width = 0.97\linewidth]{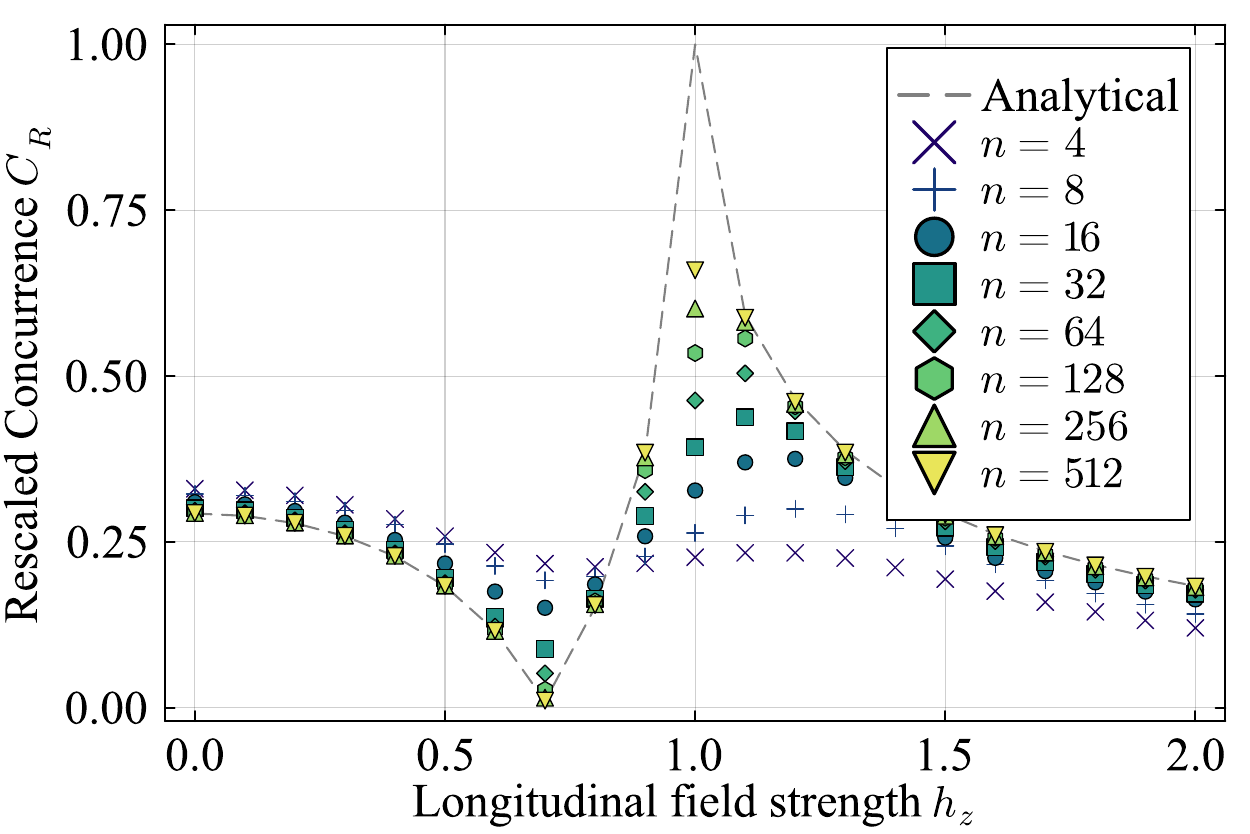}} \\
    \subfloat[$\gamma = 0.8$]{\includegraphics[width = 0.97\linewidth]{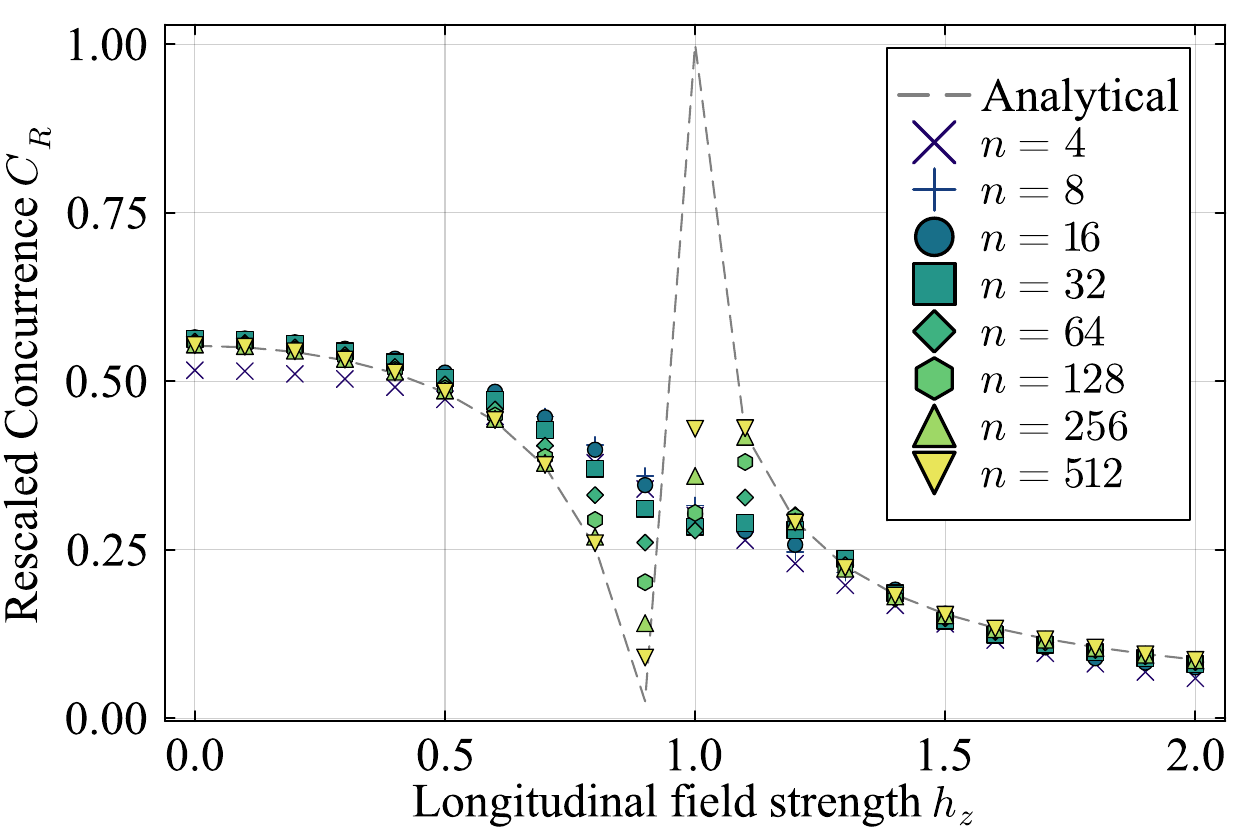}} 
    \caption{\textbf{Rescaled concurrence as a function of the longitudinal field, $h_z$.} The analytical concurrence in the thermodynamic limit is computed using Eq.~\eqref{eq:CR_th} with $J = 1$. The numerically simulated results converge toward the analytical prediction as the number of qubits $n$ increases for both (a) $\gamma = 0.5$ and (b) $\gamma = 0.8$, with the largest deviation from the thermodynamic limit value around $h_z = 1$, where the phase transition takes place. }
    \label{fig:rescaled_conc}
\end{figure}

To further characterize the properties of the prepared states, we also compute their entanglement properties as quantified through the two-qubit concurrence given by 
\begin{equation}
    C  = \max(0, \sqrt{\nu_1} - \sqrt{\nu_2} - \sqrt{\nu_3} - \sqrt{\nu_4})\,.  
\end{equation}
Here, $\nu_1 \ge \nu_2 \ge \nu_3 \ge \nu_4$ are the eigenvalues of the matrix
$R = \rho_{ij}\tilde{\rho}_{ij}$, 
where $\rho_{ij}$ denotes the reduced density matrix of qubits $i$ and $j$, and
$\tilde{\rho}_{ij} = (Y \otimes Y)\rho_{ij}^*(Y \otimes Y)$.
Due to permutation invariance, all qubit pairs are equivalent, and the concurrence is identical for any choice of $(i,j)$. In particular, we are interested in rescaled concurrence defined as
\begin{equation}
    C_R = (n - 1) C \;.
\end{equation}
In the thermodynamic limit, the rescaled concurrence admits a closed-form analytical expression~\cite{vidal2006concurrence} that depends only on $\gamma$ and $h_z$ as 
\begin{equation}
    C_R = \begin{cases}
        1 -  \sqrt{\frac{h_z - 1}{h_z - \gamma}} & \text{for } 1 \le h_z\;, \\ 
        1 -  \sqrt{\frac{1 - h_z^2}{1 - \gamma}} & \text{for } \sqrt{\gamma }\le h_z \le 1\;, \\ 
        1 -  \sqrt{\frac{1 - \gamma}{1 - h_z^2}} & \text{for } h_z \le \sqrt{\gamma}\;.
    \end{cases}
\label{eq:CR_th}
\end{equation}

Fig.~\ref{fig:rescaled_conc} compares the numerical results with the analytical prediction as a function of $h_z$ for different anisotropy parameters, with $J = 1$ fixed throughout. We clearly observe convergence as the system size increases, although the largest deviations occur near the critical point at $h_z = 1$, potentially indicating that the adiabatic procedure requires more steps. Given that our goal is to showcase our simulation algorithms and not optimize adiabatic schedules, we leave such an exploration for further work.

Finally, in Fig.~\ref{fig:runtime}, we numerically verify the classical time complexity scaling stated in Theorem~\ref{thm:time_complexity} by reporting the runtime required to simulate a single adiabatic evolution, averaged over 21 different values of the longitudinal field strength $h_z$.  Here, the upper bound corresponds to $\order{L ( d_{\lm = (n, 0)} )^{\omega}}$ which reduces to the worst-case scaling $\order{n^4}$ when taking $\omega = 3$, since the dynamics is confined to the fully symmetric sector of dimension $d_{\lambda = (n, 0)} = n+1$ and the number of time steps scales as $L \in \order{n}$.   The results clearly demonstrate that the total simulation time scales more favorably than the worst-case scenario of the theorem.

\begin{figure}[h]
    \centering
    \includegraphics[width=0.97\linewidth]{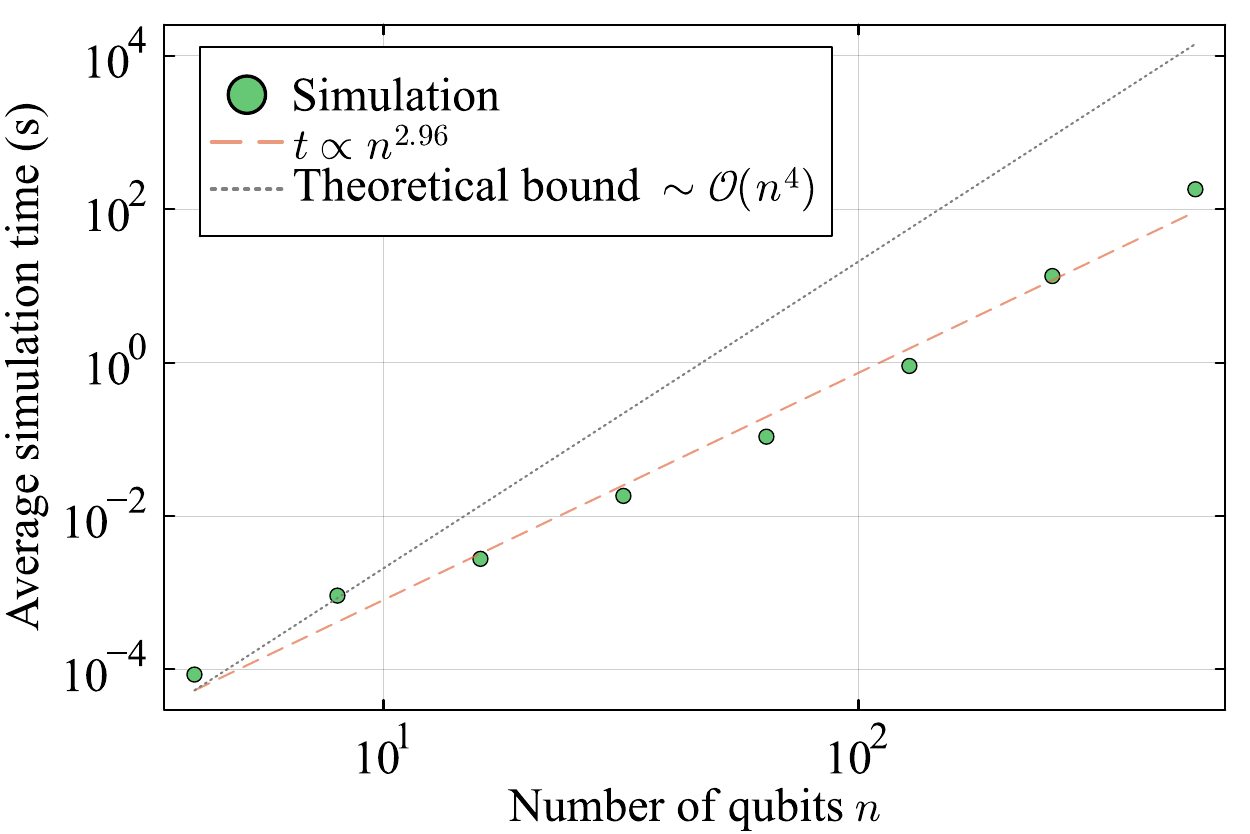}
    \caption{\textbf{Runtime of digitized AQC as a function of the number of qubits.} In this setting, the circuit depth scales as the number of Trotter steps, therefore, linearly with the number of qubits, i.e., $L \in \order{n}$, which results in the theoretical bound scaling as $\order{n^4}$ when the dynamics are confined in the fully symmetric sector. The measured runtime exhibits better scaling behavior than the theoretical upper bound.}
    \label{fig:runtime}
\end{figure}

\section{Discussion\label{sec:discussion}}
Understanding the classical simulability of families of quantum circuits is fundamental to determining the boundary between what is classical and what is genuinely quantum. Here, the gold standard is obtaining classical algorithms which scale polynomially with the system size. However, as any practitioner can attest, ``polynomial scaling'' does not always imply practical tractability. Indeed, even modest improvements in polynomial scaling and prefactors can determine whether a given circuit is simulable, or not, in practice. This is precisely the goal of our paper, as we have presented a more efficient algorithm to simulate physically motivated $S_n$-equivariant quantum circuits. Not only is our reported worst-case scaling better than that of previous techniques, but the actual clock-time obtained from the practical implementation of our algorithm shows that this bound is loose, and that the scaling is much more favorable than that of our theorems.   

Moving forward, we present two future research directions. The first is to extend our methods to $S_n$-equivariant systems of qudits with $d>2$. In this setting, the irrep label $\lm$ is partitioned into $d$ parts, $\lm = (\lm_1,\dots, \lm_d)$, which can make it harder to construct explicit Schur basis states. We expect that such a change will increase the time and memory complexity. Here, we also note that given an initial state, its irrep description can still be easily acquired using PI-CS~\cite{sauvage2024classical}.

Next, we leave for future work the comparison of physical wall-clock execution time between our classical simulation algorithm and the corresponding implementation on a quantum computer. Here, we expect that quantum deployment of $S_n$-equivariant quantum circuits could be difficult in most architectures, as one needs to create all-to-all entanglement among the qubits. If entangling gates must be executed sequentially, this could incur deeper circuits and longer wall-clock times. However,  some architectures, such as those based on trapped ions, can exploit  all-to-all interaction via shared motional modes---often referred to as a quantum bus---together with M\o lmer--S\o rensen type entangling gate~\cite{molmer1999multiparticle}. Such architectures could potentially be more amenable for $S_n$-equivariant circuits. Hence, one can readily see that a fair comparison between classical and quantum simulation run-times will ultimately be hardware dependent. Still, we expect that $S_n$-equivariant circuits could constitute a practical case where shadows plus classical simulation could scale more favorably than most quantum implementations~\cite{cerezo2023does}.

\section*{Acknowledgment}
We thank Zo\"e Holmes and Fr\'ed\'eric Sauvage for valuable discussions. 
 S.Y.C. and M.C. were supported by Laboratory Directed Research and Development (LDRD) program of Los Alamos National Laboratory (LANL) under project number 20260043DR and by the U.S. Department of Energy, Office of Science, Office of Advanced Scientific Computing Research under Contract No. DE-AC05-00OR22725 through the Accelerated Research in Quantum Computing Program MACH-Q project. M.L. and M.C. acknowledge support from LANL's ASC Beyond Moore’s Law project. This work was also supported by the Quantum Science Center (QSC), a National Quantum Information Science Research Center of the U.S. Department of Energy (DOE).

\bibliography{quantum}

@article{wiersema2023classification,
  title={Classification of dynamical Lie algebras of 2-local spin systems on linear, circular and fully connected topologies},
  author={Wiersema, Roeland and K{\"o}kc{\"u}, Efekan and Kemper, Alexander F and Bakalov, Bojko N},
  journal={npj Quantum Information},
  volume={10},
  number={1},
  pages={110},
  year={2024},
  publisher={Nature Publishing Group UK London},
url={https://www.nature.com/articles/s41534-024-00900-2},
doi={10.1038/s41534-024-00900-2}
}

@article{gibbs2024exploiting,
  title={Exploiting symmetries in nuclear Hamiltonians for ground state preparation},
  author={Gibbs, Joe and Holmes, Zo{\"e} and Stevenson, Paul},
  journal={arXiv preprint arXiv:2402.10277},
  year={2024},
  url = {https://arxiv.org/abs/2402.10277},
  doi = {10.48550/arXiv.2402.10277}
}

@article{albash2018adiabatic,
  title = {Adiabatic quantum computation},
  author = {Albash, Tameem and Lidar, Daniel A.},
  journal = {Rev. Mod. Phys.},
  volume = {90},
  issue = {1},
  pages = {015002},
  numpages = {64},
  year = {2018},
  month = {Jan},
  publisher = {American Physical Society},
  doi = {10.1103/RevModPhys.90.015002},
  url = {https://link.aps.org/doi/10.1103/RevModPhys.90.015002}
}

@book{gottesman1997stabilizer,
  title={Stabilizer codes and quantum error correction},
  author={Gottesman, Daniel},
  year={1997},
  publisher={California Institute of Technology}
}

@article{ermakov2024unified,
  title={Unified framework for efficiently computable quantum circuits},
  author={Ermakov, Igor and Lychkovskiy, Oleg and Byrnes, Tim},
  journal={arXiv preprint arXiv:2401.08187},
  year={2024},
  url = {https://arxiv.org/abs/2401.08187},
  doi = {10.48550/arXiv.2401.08187}
}

@article{bravyi2002fermionic,
  title={Fermionic quantum computation},
  author={Bravyi, Sergey B and Kitaev, Alexei Yu},
  journal={Annals of Physics},
  volume={298},
  number={1},
  pages={210--226},
  year={2002},
  publisher={Elsevier}, 
doi = {10.1006/aphy.2002.6254}, 
url = {https://doi.org/10.1006/aphy.2002.6254}
}

@article{wootters1998entanglement,
  title={Entanglement of formation of an arbitrary state of two qubits},
  author={Wootters, William K},
  journal={Physical Review Letters},
  volume={80},
  number={10},
  pages={2245},
  year={1998},
  publisher={APS},
  url={https://journals.aps.org/prl/abstract/10.1103/PhysRevLett.80.2245},
  doi={10.1103/PhysRevLett.80.2245}
}

@article{cerezo2023does,
  title={Does provable absence of barren plateaus imply classical simulability?},
  author={Cerezo, M and Larocca, Martin and Garc{\'\i}a-Mart{\'\i}n, Diego and Diaz, N L and Braccia, Paolo and Fontana, Enrico and Rudolph, Manuel S and Bermejo, Pablo and Ijaz, Aroosa and Thanasilp, Supanut and others},
 journal={Nature Communications},
  volume={16},
  number={1},
  pages={7907},
  year={2025},
  publisher={Nature Publishing Group UK London},
  url = {https://www.nature.com/articles/s41467-025-63099-6},
  doi = {10.1038/s41467-025-63099-6}
}

@article{cerezo2020variationalreview,
   title={Variational quantum algorithms},
   author={Cerezo, M. and Arrasmith, Andrew and Babbush, Ryan and Benjamin, Simon C and Endo, Suguru and Fujii, Keisuke and McClean, Jarrod R and Mitarai, Kosuke and Yuan, Xiao and Cincio, Lukasz and Coles, Patrick J. },
   journal={Nature Reviews Physics},
   volume={3},
   number={1},
   pages={625–644},
   publisher={Nature Publishing Group},
   year={2021},
   url={https://www.nature.com/articles/s42254-021-00348-9},
   doi={10.1038/s42254-021-00348-9}
 }

@article{aharonov2008adiabatic,
  title={Adiabatic quantum computation is equivalent to standard quantum computation},
  author={Aharonov, Dorit and Van Dam, Wim and Kempe, Julia and Landau, Zeph and Lloyd, Seth and Regev, Oded},
  journal={SIAM review},
  volume={50},
  number={4},
  pages={755--787},
  year={2008},
  publisher={SIAM},
url={https://www.jstor.org/stable/20454175}
}

@inproceedings{mernyei2022equivariant,
  title={Equivariant quantum graph circuits},
  author={Mernyei, P{\'e}ter and Meichanetzidis, Konstantinos and Ceylan, Ismail Ilkan},
  booktitle={International Conference on Machine Learning},
  pages={15401--15420},
  year={2022},
  organization={PMLR},
  url={https://proceedings.mlr.press/v162/mernyei22a.html}
}

@article{nguyen2022atheory,
title = {Theory for Equivariant Quantum Neural Networks},
  author = {Nguyen, Quynh T. and Schatzki, Louis and Braccia, Paolo and Ragone, Michael and Coles, Patrick J. and Sauvage, Fr\'ed\'eric and Larocca, Mart\'{\i}n and Cerezo, M.},
  journal = {PRX Quantum},
  volume = {5},
  issue = {2},
  pages = {020328},
  numpages = {40},
  year = {2024},
  month = {May},
  publisher = {American Physical Society},
  doi = {10.1103/PRXQuantum.5.020328},
  url = {https://link.aps.org/doi/10.1103/PRXQuantum.5.020328}
}

@article{knill2001fermionic,
      title={Fermionic Linear Optics and Matchgates}, 
      author={E. Knill},
      year={2001},
  url = {https://arxiv.org/abs/quant-ph/0108033},
  journal={arXiv preprint arXiv:quant-ph/0108033}, 
doi={10.48550/arXiv.quant-ph/0108033} 
}

@article{schatzki2022theoretical,
  title={Theoretical guarantees for permutation-equivariant quantum neural networks},
  author={Schatzki, Louis and Larocca, Martin and Nguyen, Quynh T and Sauvage, Frederic and Cerezo, Marco},
  journal={npj Quantum Information},
  volume={10},
  number={1},
  pages={12},
  year={2024},
  publisher={Nature Publishing Group UK London},
  url = {https://www.nature.com/articles/s41534-024-00804-1},
doi = {10.1038/s41534-024-00804-1},
}

@article{ragone2022representation,
  title={Representation Theory for Geometric Quantum Machine Learning},
  author={ Ragone,  Michael  and Nguyen, Quynh T. and Schatzki,  Louis  and  Braccia, Paolo and
Larocca, Martin and   Sauvage, Frederic and Coles,  Patrick J. and  Cerezo, M. },
  year={2022},
  url = {https://arxiv.org/abs/2210.07980},
  journal={arXiv preprint arXiv:2210.07980},
doi={10.48550/arXiv.2210.07980}
}

@article{verdon2019quantumgraph,
  title={Quantum graph neural networks},
  author={Verdon, Guillaume and McCourt, Trevor and Luzhnica, Enxhell and Singh, Vikash and Leichenauer, Stefan and Hidary, Jack},
  journal={arXiv preprint arXiv:1909.12264},
  year={2019},
  url={https://arxiv.org/abs/1909.12264}, 
doi={10.48550/arXiv.1909.12264}
}

@article{ragone2023lie,
  title={A Lie algebraic theory of barren plateaus for deep parameterized quantum circuits},
  author={Ragone, Michael and Bakalov, Bojko N and Sauvage, Fr{\'e}d{\'e}ric and Kemper, Alexander F and Ortiz Marrero, Carlos and Larocca, Mart{\'\i}n and Cerezo, M},
  journal={Nature Communications},
  volume={15},
  number={1},
  pages={7172},
  year={2024},
  publisher={Nature Publishing Group UK London},
url={https://www.nature.com/articles/s41467-024-49909-3#citeas},
doi={10.1038/s41467-024-49909-3}
}

@article{wan2022matchgate,
  title={Matchgate shadows for fermionic quantum simulation},
  author={Wan, Kianna and Huggins, William J and Lee, Joonho and Babbush, Ryan},
  journal={Communications in Mathematical Physics},
  volume={404},
  pages={629},
  year={2023},
  publisher={Springer},
url={https://link.springer.com/article/10.1007/s00220-023-04844-0},
doi={10.1007/s00220-023-04844-0}
}

@article{sauvage2024classical,
  title={Classical shadows with symmetries},
  author={Sauvage, Frederic and Larocca, Martin},
  journal={arXiv preprint arXiv:2408.05279},
  year={2024},
url={https://arxiv.org/abs/2408.05279},
doi={10.48550/arXiv.2408.05279}
}

@article{zheng2025toward,
    title = {Toward superpolynomial quantum speedup of equivariant quantum algorithms with $\mathrm{SU}(d)$ symmetry},
  author={Zheng, Han and Li, Zimu and Strelchuk, Sergii and Kondor, Risi and Liu, Junyu},
  journal={Physical Review A},
  volume={112},
  number={5},
  pages={052435},
  year={2025},
  publisher={APS}, 
  doi = {10.1103/pt27-v2nj},
  url = {https://link.aps.org/doi/10.1103/pt27-v2nj}
}

@article{sack2022avoiding,
  title={Avoiding barren plateaus using classical shadows},
  author={Sack, Stefan H and Medina, Raimel A and Michailidis, Alexios A and Kueng, Richard and Serbyn, Maksym},
  journal={PRX Quantum},
  volume={3},
  number={2},
  pages={020365},
  year={2022},
  publisher={APS},
  url={https://journals.aps.org/prxquantum/abstract/10.1103/PRXQuantum.3.020365},
  doi={10.1103/PRXQuantum.3.020365}
}

@article{kazi2023universality,
  title={On the universality of $S_n$-equivariant $ k $-body gates},
  author={Kazi, Sujay and Larocca, Martin and Cerezo, M},
doi = {10.1088/1367-2630/ad4819},
url = {https://dx.doi.org/10.1088/1367-2630/ad4819},
year = {2024},
month = {may},
publisher = {IOP Publishing},
volume = {26},
number = {5},
pages = {053030},
journal = {New Journal of Physics}
}

@inproceedings{bartschi2019deterministic,
  title={Deterministic preparation of Dicke states},
  author={B{\"a}rtschi, Andreas and Eidenbenz, Stephan},
  booktitle={International Symposium on Fundamentals of Computation Theory},
  pages={126--139},
  year={2019},
  organization={Springer},
  url={https://link.springer.com/chapter/10.1007/978-3-030-25027-0_9},
  doi={10.1007/978-3-030-25027-0_9}
}

@article{kazi2022landscape,
  title={Analyzing the quantum approximate optimization algorithm: ans{\"a}tze, symmetries, and lie algebras},
  author={Kazi, Sujay and Larocca, Mart{\'\i}n and Farinati, Marco and Coles, Patrick J and Cerezo, M and Zeier, Robert},
  journal={PRX Quantum},
  volume={6},
  number={4},
  pages={040345},
  year={2025},
  publisher={APS},
url={https://journals.aps.org/prxquantum/abstract/10.1103/yfwq-yqmk},
doi={10.1103/yfwq-yqmk}
}

@article{huang2020predicting,
  title={Predicting many properties of a quantum system from very few measurements},
  author={Huang, Hsin-Yuan and Kueng, Richard and Preskill, John},
  journal={Nature Physics},
  volume={16},
  number={10},
  pages={1050--1057},
  year={2020},
  publisher={Nature Publishing Group},
  doi={10.1038/s41567-020-0932-7},
  url={https://www.nature.com/articles/s41567-020-0932-7}
}

@book{fulton1991representation,
  title={Representation Theory: A First Course},
  author={Fulton, William and Harris, Joe},
  year={1991},
  publisher={Springer}
}

@article{jozsa2008matchgates,
  title={Matchgates and classical simulation of quantum circuits},
  author={Jozsa, Richard and Miyake, Akimasa},
  journal={Proceedings of the Royal Society A: Mathematical, Physical and Engineering Sciences},
  volume={464},
  number={2100},
  pages={3089--3106},
  year={2008},
  publisher={The Royal Society London},
url={https://royalsocietypublishing.org/doi/10.1098/rspa.2008.0189},
doi={10.1098/rspa.2008.0189}
}

@article{gottesman1998heisenbergrepresentation,
  title={The Heisenberg representation of quantum computers},
  author={Gottesman, Daniel},
  journal={arXiv preprint quant-ph/9807006},
  year={1998},
  url={https://arxiv.org/abs/quant-ph/9807006}, 
 doi={10.48550/arXiv.quant-ph/9807006}
}

@article{cirstoiu2024fourier,
  title={A Fourier analysis framework for approximate classical simulations of quantum circuits},
  author={Cirstoiu, Cristina},
  journal={arXiv preprint arXiv:2410.13856},
  year={2024},
  url={https://arxiv.org/abs/2410.13856},
doi={10.48550/arXiv.2410.13856}
}

@article{skolik2022equivariant,
  title={Equivariant quantum circuits for learning on weighted graphs},
  author={Skolik, Andrea and Cattelan, Michele and Yarkoni, Sheir and B{\"a}ck, Thomas and Dunjko, Vedran},
  journal={npj Quantum Information},
  volume={9},
  number={1},
  pages={47},
  year={2023},
  publisher={Nature Publishing Group UK London},
  url={https://www.nature.com/articles/s41534-023-00710-y},
  doi={10.1038/s41534-023-00710-y}
}

@article{kirby2017practical,
title={A practical quantum algorithm for the Schur transform},
author = {Kirby, William M. and Strauch, Frederick W.},
year = {2018},
journal={Quantum Information \& Computation},
publisher = {Rinton Press, Incorporated},
address = {Paramus, NJ},
volume = {18},
number = {9–10},
pages = {721–742},
url={https://dl.acm.org/doi/10.5555/3370214.3370215},
doi={10.5555/3370214.3370215}
}

@article{krovi2019efficient,
  title={An efficient high dimensional quantum Schur transform},
  author={Krovi, Hari},
  journal={Quantum},
  volume={3},
  pages={122},
  year={2019},
  publisher={Verein zur F{\"o}rderung des Open Access Publizierens in den Quantenwissenschaften},
  url={https://quantum-journal.org/papers/q-2019-02-14-122/},
  doi={10.22331/q-2019-02-14-122}
}

@article{Vidal2003Efficient,
  title = {Efficient Classical Simulation of Slightly Entangled Quantum Computations},
  author = {Vidal, Guifr\'e},
  journal = {Phys. Rev. Lett.},
  volume = {91},
  issue = {14},
  pages = {147902},
  numpages = {4},
  year = {2003},
  month = {Oct},
  publisher = {American Physical Society},
  doi = {10.1103/PhysRevLett.91.147902},
  url = {https://link.aps.org/doi/10.1103/PhysRevLett.91.147902}
}

@article{elben2022randomized,
  title={The randomized measurement toolbox},
  author={Elben, Andreas and Flammia, Steven T and Huang, Hsin-Yuan and Kueng, Richard and Preskill, John and Vermersch, Beno{\^\i}t and Zoller, Peter},
  journal={Nature Reviews Physics},
  volume={5},
  number={1},
  pages={9--24},
  year={2023},
  publisher={Nature Publishing Group UK London}, 
    doi={10.1038/s42254-022-00535-2},
  url={https://www.nature.com/articles/s42254-022-00535-2}

}

@article{anschuetz2022efficient,
  title={Efficient classical algorithms for simulating symmetric quantum systems},
  author={Anschuetz, Eric R and Bauer, Andreas and Kiani, Bobak T and Lloyd, Seth},
  journal={Quantum},
  volume={7},
  pages={1189},
  year={2023},
  publisher={Verein zur F{\"o}rderung des Open Access Publizierens in den Quantenwissenschaften},
url={https://quantum-journal.org/papers/q-2023-11-28-1189/},
doi={10.22331/q-2023-11-28-1189}
}

@article{childs2021theory,
  title={Theory of trotter error with commutator scaling},
  author={Childs, Andrew M and Su, Yuan and Tran, Minh C and Wiebe, Nathan and Zhu, Shuchen},
  journal={Physical Review X},
  volume={11},
  number={1},
  pages={011020},
  year={2021},
  publisher={APS},
  doi={https://doi.org/10.1103/PhysRevX.11.011020}
}

@article{bacon2006efficient,
  title={Efficient quantum circuits for Schur and Clebsch-Gordan transforms},
  author={Bacon, Dave and Chuang, Isaac L and Harrow, Aram W},
  journal={Physical review letters},
  volume={97},
  number={17},
  pages={170502},
  year={2006},
  publisher={APS},
  doi = {10.1103/PhysRevLett.97.170502},
  url = {https://link.aps.org/doi/10.1103/PhysRevLett.97.170502}
}

@article{east2023all,
  title={All you need is spin: SU (2) equivariant variational quantum circuits based on spin networks},
  author={East, Richard DP and Alonso-Linaje, Guillermo and Park, Chae-Yeun},
  journal={arXiv preprint arXiv:2309.07250},
  year={2023},
url={https://arxiv.org/abs/2309.07250}, 
doi={10.48550/arXiv.2309.07250 }
}

@article{bharti2022noisy,
  title={Noisy intermediate-scale quantum algorithm for semidefinite programming},
  author={Bharti, Kishor and Haug, Tobias and Vedral, Vlatko and Kwek, Leong-Chuan},
  journal={Physical Review A},
  volume={105},
  number={5},
  pages={052445},
  year={2022},
  publisher={APS},
  doi = {10.1103/PhysRevA.105.052445},
  url = {https://doi.org/10.1103/PhysRevA.105.052445}
}

@article{feng2025quon,
  title={Quon Classical Simulation: Unifying Clifford, Matchgates and Entanglement},
  author={Feng, Zixuan and Liu, Zhengwei and Lu, Fan and Wang, Ningfeng},
  journal={arXiv preprint arXiv:2505.07804},
  year={2025},
url={https://arxiv.org/abs/2505.07804},
doi={10.48550/arXiv.2505.07804}
}

@article{bacon2005quantum,
  title={The quantum Schur transform: I. efficient qudit circuits},
  author={Bacon, Dave and Chuang, Isaac L and Harrow, Aram W},
  journal={arXiv preprint arXiv:0601001},
  year={2005},
  url={https://arxiv.org/abs/quant-ph/0601001}, 
doi={10.48550/arXiv.quant-ph/0601001}
}

@article{bezanson2017julia,
    title={Julia: A fresh approach to numerical computing},
    author={Bezanson, Jeff and Edelman, Alan and Karpinski, Stefan and Shah, Viral B},
    journal={SIAM {R}eview},
    volume={59},
    number={1},
    pages={65--98},
    year={2017},
    publisher={SIAM},
    doi={10.1137/141000671},
    url={https://epubs.siam.org/doi/10.1137/141000671}
}

@article{bravyi2016improved,
  title = {Improved Classical Simulation of Quantum Circuits Dominated by Clifford Gates},
  author = {Bravyi, Sergey and Gosset, David},
  journal = {Phys. Rev. Lett.},
  volume = {116},
  issue = {25},
  pages = {250501},
  numpages = {5},
  year = {2016},
  month = {Jun},
  publisher = {American Physical Society},
  doi = {10.1103/PhysRevLett.116.250501},
  url = {https://link.aps.org/doi/10.1103/PhysRevLett.116.250501}
}

@article{angrisani2025simulating,
  title={Simulating quantum circuits with arbitrary local noise using Pauli Propagation},
  author={Angrisani, Armando and Mele, Antonio A and Rudolph, Manuel S and Cerezo, M and Holmes, Zoe},
  journal={arXiv preprint arXiv:2501.13101},
  year={2025},
  url={https://arxiv.org/abs/2501.13101},
  doi ={10.48550/arXiv.2501.13101}
}

@article{angrisani2024classically,
  title = {Classically Estimating Observables of Noiseless Quantum Circuits},
  author = {Angrisani, Armando and Schmidhuber, Alexander and Rudolph, Manuel S. and Cerezo, M. and Holmes, Zo\"e and Huang, Hsin-Yuan},
  journal = {Phys. Rev. Lett.},
  volume = {135},
  issue = {17},
  pages = {170602},
  numpages = {10},
  year = {2025},
  month = {Oct},
  publisher = {American Physical Society},
  doi = {10.1103/lh6x-7rc3},
  url = {https://link.aps.org/doi/10.1103/lh6x-7rc3}
}

@article{meyer2023exploiting,
  title={Exploiting symmetry in variational quantum machine learning},
  author={Meyer, Johannes Jakob and Mularski, Marian and Gil-Fuster, Elies and Mele, Antonio Anna and Arzani, Francesco and Wilms, Alissa and Eisert, Jens},
  journal={PRX Quantum},
  volume={4},
  number={1},
  pages={010328},
  year={2023},
  publisher={APS},
  url={https://doi.org/10.1103/PRXQuantum.4.010328},
  doi={10.1103/PRXQuantum.4.010328}
}

@article{zhao2021fermionic,
  title={Fermionic partial tomography via classical shadows},
  author={Zhao, Andrew and Rubin, Nicholas C and Miyake, Akimasa},
  journal={Physical Review Letters},
  volume={127},
  number={11},
  pages={110504},
  year={2021},
  publisher={APS}, 
  doi = {10.1103/PhysRevLett.127.110504},
  url = {https://link.aps.org/doi/10.1103/PhysRevLett.127.110504}
}

@article{bertoni2024shallow,
  title={Shallow shadows: Expectation estimation using low-depth random {C}lifford circuits},
  author={Bertoni, Christian and Haferkamp, Jonas and Hinsche, Marcel and Ioannou, Marios and Eisert, Jens and Pashayan, Hakop},
  journal={Physical Review Letters},
  volume={133},
  number={2},
  pages={020602},
  year={2024},
  publisher={APS},
url={https://journals.aps.org/prl/abstract/10.1103/PhysRevLett.133.020602},
doi={10.1103/PhysRevLett.133.020602}
}

@article{van2022hardware,
  title={Hardware-efficient learning of quantum many-body states},
  author={Van Kirk, Katherine and Cotler, Jordan and Huang, Hsin-Yuan and Lukin, Mikhail D},
  journal={arXiv preprint arXiv:2212.06084},
  year={2022},
  url={https://arxiv.org/abs/2212.06084}, 
doi={10.48550/arXiv.2212.06084}
}

@incollection{howe2022irreducible,
  title={The Irreducible Representations of S n: Young Symmetrizers},
  author={Howe, R Michael},
  booktitle={An Invitation to Representation Theory: Polynomial Representations of the Symmetric Group},
  pages={103--123},
  year={2022},
  publisher={Springer}, 
  doi={10.1007/978-3-030-98025-2_8},
  url={https://doi.org/10.1007/978-3-030-98025-2_8}
}

@article{holtz2012alternating,
  title={The alternating linear scheme for tensor optimization in the tensor train format},
  author={Holtz, Sebastian and Rohwedder, Thorsten and Schneider, Reinhold},
  journal={SIAM Journal on Scientific Computing},
  volume={34},
  number={2},
  pages={A683--A713},
  year={2012},
  publisher={SIAM},
  url = {https://epubs.siam.org/doi/abs/10.1137/100818893},
  doi = {10.1137/100818893}
}

@article{west2024real,
  title={Real classical shadows},
  author={West, Maxwell and Mele, Antonio Anna and Larocca, Martin and Cerezo, M},
  journal={arXiv preprint arXiv:2410.23481},
  year={2024},
  url={https://arxiv.org/abs/2410.23481}, 
doi={10.48550/arXiv.2410.23481} 
}

@article{brod2016efficient,
  title={Efficient classical simulation of matchgate circuits with generalized inputs and measurements},
  author={Brod, Daniel J.},
  journal={Physical Review A},
  volume={93},
  number={6},
  pages={062332},
  year={2016},
  publisher={APS},
  doi={10.1103/PhysRevA.93.062332},
  url={https://journals.aps.org/pra/abstract/10.1103/PhysRevA.93.062332}
}

@article{chang2025primer,
  title={A Primer on Quantum Machine Learning},
  author={Chang, Su Yeon and Cerezo, M},
  journal={arXiv preprint arXiv:2511.15969},
  year={2025},
url={https://arxiv.org/abs/2511.15969},
doi={10.48550/arXiv.2511.15969
}
}

@article{kokcu2024classification,
  title={Classification of dynamical Lie algebras generated by spin interactions on undirected graphs},
  author={K{\"o}kc{\"u}, Efekan and Wiersema, Roeland and Kemper, Alexander F. and Bakalov, Bojko N.},
  journal={arXiv preprint arXiv:2409.19797},
  year={2024},
  doi={10.48550/arXiv.2409.19797},
  url={https://arxiv.org/abs/2409.19797}
}

@book{fulton1997young,
  title={Young tableaux: with applications to representation theory and geometry},
  author={Fulton, William},
  number={35},
  year={1997},
  publisher={Cambridge University Press}
}

@article{li2024enforcing,
  title={Enforcing exact permutation and rotational symmetries in the application of quantum neural networks on point cloud datasets},
  author={Li, Zhelun and Nagano, Lento and Terashi, Koji},
  journal={Physical Review Research},
  volume={6},
  number={4},
  pages={043028},
  year={2024},
  publisher={APS}, 
  doi = {10.1103/PhysRevResearch.6.043028},
  url = {https://link.aps.org/doi/10.1103/PhysRevResearch.6.043028}
}

@article{cervia2021lipkin,
  title={Lipkin model on a quantum computer},
  author={Cervia, Michael J and Balantekin, AB and Coppersmith, SN and Johnson, Calvin W and Love, Peter J and Poole, C and Robbins, K and Saffman, M},
  journal={Physical Review C},
  volume={104},
  number={2},
  pages={024305},
  year={2021},
  publisher={APS},
  url={https://journals.aps.org/prc/abstract/10.1103/PhysRevC.104.024305},
  doi={10.1103/PhysRevC.104.024305}
}

@article{pal2023complexity,
  title={Complexity in the Lipkin-Meshkov-Glick model},
  author={Pal, Kunal and Pal, Kuntal and Sarkar, Tapobrata},
  journal={Physical Review E},
  volume={107},
  number={4},
  pages={044130},
  year={2023},
  publisher={APS},
  doi = {10.1103/PhysRevE.107.044130},
  url = {https://link.aps.org/doi/10.1103/PhysRevE.107.044130}
}

@article{yadin2023thermodynamics,
  title={Thermodynamics of permutation-invariant quantum many-body systems: A group-theoretical framework},
  author={Yadin, Benjamin and Morris, Benjamin and Brandner, Kay},
  journal={Physical Review Research},
  volume={5},
  number={3},
  pages={033018},
  year={2023},
  publisher={APS}, 
  doi = {10.1103/PhysRevResearch.5.033018},
  url = {https://link.aps.org/doi/10.1103/PhysRevResearch.5.033018}
}

@article{kraus2013ground,
  title={Ground states of fermionic lattice Hamiltonians with permutation symmetry},
  author={Kraus, Christina V and Lewenstein, Maciej and Cirac, J Ignacio},
  journal={Physical Review A—Atomic, Molecular, and Optical Physics},
  volume={88},
  number={2},
  pages={022335},
  year={2013},
  publisher={APS}, 
  doi = {10.1103/PhysRevA.88.022335},
  url = {https://link.aps.org/doi/10.1103/PhysRevA.88.022335}
}

@article{huang2021dynamic,
  title={Dynamic synthesis of Heisenberg-limited spin squeezing},
  author={Huang, Long-Gang and Chen, Feng and Li, Xinwei and Li, Yaohua and L{\"u}, Rong and Liu, Yong-Chun},
  journal={npj Quantum Information},
  volume={7},
  number={1},
  pages={168},
  year={2021},
  publisher={Nature Publishing Group UK London}, 
doi={https://doi.org/10.1038/s41534-021-00505-z}, 
url={https://www.nature.com/articles/s41534-021-00505-z}
}

@article{zhang2024heisenberg,
  title={Heisenberg-limit spin squeezing with the spin Bogoliubov Hamiltonian},
  author={Zhang, Jun and Chang, Sheng and Zhang, Wenxian},
  journal={Physical Review A},
  volume={110},
  number={3},
  pages={033704},
  year={2024},
  publisher={APS}, 
  doi = {10.1103/PhysRevA.110.033704},
  url = {https://link.aps.org/doi/10.1103/PhysRevA.110.033704}
}

@article{dusuel2005continuous,
  title={Continuous unitary transformations and finite-size scaling exponents in the Lipkin-Meshkov-Glick model},
  author={Dusuel, S{\'e}bastien and Vidal, Julien},
  journal={Physical Review B—Condensed Matter and Materials Physics},
  volume={71},
  number={22},
  pages={224420},
  year={2005},
  publisher={APS}, 
  doi = {10.1103/PhysRevB.71.224420},
  url = {https://link.aps.org/doi/10.1103/PhysRevB.71.224420}
}

@article{dusuel2004finite,
  title={Finite-size scaling exponents of the Lipkin-Meshkov-Glick model},
  author={Dusuel, S{\'e}bastien and Vidal, Julien},
  journal={Physical review letters},
  volume={93},
  number={23},
  pages={237204},
  year={2004},
  publisher={APS}, 
  doi = {10.1103/PhysRevLett.93.237204},
  url = {https://link.aps.org/doi/10.1103/PhysRevLett.93.237204}
}

@article{vidal2006concurrence,
  title={Concurrence in collective models},
  author={Vidal, Julien},
  journal={Physical Review A—Atomic, Molecular, and Optical Physics},
  volume={73},
  number={6},
  pages={062318},
  year={2006},
  publisher={APS}, 
  doi = {10.1103/PhysRevA.73.062318},
  url = {https://link.aps.org/doi/10.1103/PhysRevA.73.062318}
}

@article{teng2025leveraging,
  title={Leveraging Symmetry Merging in Pauli Propagation},
  author={Teng, Yanting and Chang, Su Yeon and Rudolph, Manuel S and Holmes, Zo{\"e}},
  journal={arXiv preprint arXiv:2512.12094},
  year={2025}, 
url={https://doi.org/10.48550/arXiv.2512.12094}, 
doi={10.48550/arXiv.2512.12094}
}

@article{barends2016digitized,
  title={Digitized adiabatic quantum computing with a superconducting circuit},
  author={Barends, Rami and Shabani, Alireza and Lamata, Lucas and Kelly, Julian and Mezzacapo, Antonio and Heras, U Las and Babbush, Ryan and Fowler, Austin G and Campbell, Brooks and Chen, Yu and others},
  journal={Nature},
  volume={534},
  number={7606},
  pages={222--226},
  year={2016},
  publisher={Nature Publishing Group UK London}, 
doi={10.1038/nature17658}, 
url={https://doi.org/10.1038/nature17658}
}

@article{born1928beweis,
  title={Beweis des adiabatensatzes},
  author={Born, Max and Fock, Vladimir},
  journal={Zeitschrift f{\"u}r Physik},
  volume={51},
  number={3},
  pages={165--180},
  year={1928},
  publisher={Springer}, 
url={https://link.springer.com/article/10.1007/BF01343193}, 
doi={10.1007/BF01343193}
}

@article{lipkin1965validity,
  title={Validity of many-body approximation methods for a solvable model:(I). Exact solutions and perturbation theory},
  author={Lipkin, Harry J and Meshkov, N and Glick, AJ},
  journal={Nuclear Physics},
  volume={62},
  number={2},
  pages={188--198},
  year={1965},
  publisher={Elsevier}, 
doi={10.1016/0029-5582(65)90862-X}, 
url={https://www.sciencedirect.com/science/article/pii/002955826590862X}
}

@article{molmer1999multiparticle,
  title={Multiparticle entanglement of hot trapped ions},
  author={M{\o}lmer, Klaus and S{\o}rensen, Anders},
  journal={Physical Review Letters},
  volume={82},
  number={9},
  pages={1835},
  year={1999},
  publisher={APS}, 
  doi = {10.1103/PhysRevLett.82.1835},
  url = {https://link.aps.org/doi/10.1103/PhysRevLett.82.1835}
}

@article{park2026hyqurp,
  title={HyQuRP: Hybrid quantum-classical neural network with rotational and permutational equivariance for 3D point clouds},
  author={Park, Semin and Park, Chae-Yeun},
  journal={arXiv preprint arXiv:2602.06381},
  year={2026}, 
  doi={10.48550/arXiv.2602.06381}, 
  url={https://arxiv.org/abs/2602.06381}
}

\clearpage
\newpage
\onecolumngrid
\appendix

\section*{Appendices}

Here, we present additional details, proofs, and further descriptions of the algorithms used in the main text.

\section{Explicit expression for single-body, two-body, and global Pauli operators\label{adx:explicit_blocks}}
In this appendix, we present explicit Schur basis matrix representations, $\left(A_\lm\right)_{q_\lm q'_\lm} = \bramatket{\lm, p^0_\lm, q_\lm}{A}{\lm, p^0_\lm, q_\lm'}$ for several $S_n$-equivariant Pauli operators $A$ using the canonical Schur basis states $\ket{\lm, p^0_\lm , q_\lm}$ introduced in Section~\ref{sec:schur_basis}. We focus on operators of the form 
\begin{equation} 
    \frac{1}{n}\sum^n_{i=1} P_j, \quad \frac{2}{n(n-1)}\sum_{1\leq k<j\leq n } P_j  P_k,  \quad \bigotimes_{j =1}^n P_j\,,
\end{equation}
where $P_j\in \{X,Y,Z\}$ acts on the $j$-th qubit. 

We begin by presenting the Schur basis matrix elements of symmetrized Pauli strings composed only of Pauli-$X$ operators. By analyzing the action of the single-qubit operator $X$ and the two-qubit operator $X\otimes X$ on the antisymmetric and symmetric components of the Schur basis,  the matrix elements are explicitly computed. For $n \ge 3$, they are given by
\small
\begin{equation}
    \left(A_\lm\right)_{q_\lm q'_\lm}  = 
    \begin{cases}       \frac{2\left(q_\lm(n- 2m - q_\lm) -m \right)}{n(n-1)}  \dl_{q_\lm, q'_\lm}   
+ \frac{\al^{-}(\lm, q_\lm) \al^{-}(\lm, q_\lm - 1) }{n(n-1)}\dl_{q_\lm-2, q'_\lm} 
     + \frac{\al^{+}(\lm, q_\lm) \al^{+}(\lm, q_\lm+1)}{n(n-1)}\dl_{q_\lm + 2ß, q'_\lm} 
     
     & \text{if } A = \frac{2}{n(n-1)}\sum_{k < j} X_j X_k\;, \\  
        \frac{1}{n }\left(\al^{-}(\lm, q_\lm) \dl_{q_\lm -1, q'_\lm}  + \al^{+}(\lm, q_\lm)\dl_{q_\lm + 1, q'_\lm}\right) 
        & \text{if } A = \frac{1}{n}\sum_{i} X_i\;, \\  
        (-1)^{m}\dl_{q_\lm, n - 2m - q'_\lm}\;, & \text{if } A = X^{\otimes n}\;.
    \end{cases}\nonumber
\end{equation}
\normalsize
Each irrep block is therefore sparse, with non-zero entries confined only to the main diagonal and at most the second super- and sub-diagonals.
Therefore, each block $A_\lm$ contains at most $\order{3n}$ non-zero entries, which results in $\order{n^2}$ operations for matrix construction.

The Schur basis representations for Pauli-$Y$ follow analogously. Using that $Y = iXZ$, the corresponding matrix elements are
\small 
\begin{equation}
    \left(A_\lm\right)_{q_\lm q'_\lm}  = 
    \begin{cases}       \frac{2\left(q_\lm(n- 2m - q_\lm) -m \right)}{n(n-1)}  \dl_{q_\lm, q'_\lm}   
 - \frac{\al^{-}(\lm, q_\lm) \al^{-}(\lm, q_\lm - 1) }{n(n-1)}\dl_{q_\lm-2, q'_\lm} 
     - \frac{\al^{+}(\lm, q_\lm) \al^{+}(\lm, q_\lm+1)}{n(n-1)}\dl_{q_\lm + 2ß, q'_\lm} 
     
     & \text{if } A = \frac{2}{n(n-1)}\sum_{k < j} Y_j Y_k\;, \\  
        \frac{i}{n}\left(\al^{-}(\lm, q_\lm) \dl_{q_\lm -1, q'_\lm}  - \al^{+}(\lm, q_\lm)\dl_{q_\lm + 1, q'_\lm}\right) 
        & \text{if } A = \frac{1}{n}\sum_{i} Y_i\;, \\  
        \left(-1\right)^{q_\lm + \frac{n}{2}}\dl_{q_\lm, n - 2m - q'_\lm} \;, & \text{if } A = Y^{\otimes n}\;.\nonumber
    \end{cases}
\end{equation}
\normalsize

Finally, as described in Section~\ref{sec:Sn_operators}, Pauli-$Z$ operators are diagonal in the Schur basis. With the shortened notation $h_{\lm,q_\lm} = q_\lm + m$ for the total Hamming weight of the Schur basis $\ket{\lm, p_\lm, q_\lm}$,  the matrix elements are explicitly given as:
\begin{equation}
    \left(A_\lm\right)_{q_\lm q'_\lm}  = 
    \begin{cases}        \frac{1}{n(n-1)}\left(n^2 -n - 4nh_{\lm,q_\lm} + 4h_{\lm,q_\lm}^2\right)\dl_{q_\lm, q'_\lm} & \text{if } A = \frac{2}{n(n-1)}\sum_{k < j} Z_j Z_k\;, \\  
        \frac{1}{n}\left(n - 2 h_{\lm,q_\lm}\right) \dl_{q_\lm, q'_\lm}& \text{if } A = \frac{1}{n}\sum_{i} Z_i\;, \\  
        (-1)^{h_{\lm,q_\lm}} \dl_{q_\lm, q'_\lm} & \text{if } A = Z^{\otimes n}\;.\nonumber
    \end{cases}
\end{equation}

Using the same Schur basis projection procedure, the matrix can be computed for any arbitrary two-local Pauli operators $\TC_{S_n}(P\otimes Q)$, $P, Q\in \{X, Y, Z, I\}$, which are employed, for example, in the computation of two-qubit reduced states in the numerical simulation of Section~\ref{sec:numerics}.

\section{Calculating matrix elements for the $k$-local symmetrized Pauli strings~\label{adx:k_local_generators}} 
In Section~\ref{sec:k_local_PS} of the main text, we presented a theorem stating the computational complexity of evaluating matrix elements for arbitrary $k$-local symmetrized Pauli operators given that $k \in \order{1}$. In this appendix, we justify this claim by presenting an explicit numerical algorithm for computing these matrix elements and by analyzing its computational complexity. Although the resulting expressions do not admit a simple closed form, the algorithm enables efficient evaluation of the matrix elements. Our derivation and complexity analysis are inspired by the approach developed in Appendix~E of Ref.~\cite{anschuetz2022efficient}.

\begin{theorem}[Action of symmetrized Pauli strings in the Schur basis]
Let $P_{\vec{k}}$  be a Pauli string acting on $n$ qubits as defined in Eq.~\eqref{eq:Pk} for $\vec{k} = (k_X, k_Y, k_Z)$ and locality $|\vec{k}|=k$, and let $\TC_{S_n}\!\left(P_{\vec{k}}\right)$  denote the symmetrized Pauli string. Consider the canonical Schur basis $\ket{\lm, p_\lm^0, q_\lm}$ defined by Eq.~\eqref{eq:schur_basis_expression}, which contains $m$ antisymmetric singlet state pairs and the symmetric Dicke state of Hamming weight $q_\lm$. Then, the action of $\TC_{S_n}\!\left(P_{\vec{k}}\right)$ on $\ket{\lm, p_\lm^0, q_\lm}$ is given by
\begin{align}
\TC_{S_n}\!\left(P_{\vec{k}}\right)\ket{\lambda, p_\lambda^0, q_\lambda}
&=
\frac{k_X! k_Y!k_Z!(n-k)!}{n!}
\sum_{\substack{\vec{a} \in\mathbb{Z}_{\ge 0}^3\\ \vec{a}\le \left\lfloor\frac{\vec{k}}{2}\right\rfloor\\ \norm{\vec{a}}_1 \le m}}
\ \sum_{\substack{\vec{s}\in\mathbb{Z}_{\ge 0}^3\\ \vec{s}\le \vec{k}- 2\vec{a}\\ \norm{\vec{s}}_1 \le q_\lambda}}
\mathcal{N}(q_\lm, q_\lm')
c(k_Y, \vec{a}, \vec{s}) W(q_\lm, \vec{k}, \vec{a}, \vec{s} ) 
\ket{\lambda,\; p_\lambda^0, q_\lm'(\vec{k}, \vec{a}, \vec{s})}\;, 
\label{eq:action_schur}
\end{align}
with the output Hamming weight defined as $q_\lm \equiv q_\lm'(\vec{k}, \vec{a}, \vec{s}) = q_\lambda + k_X + k_Y - 2a_X - 2a_Y - 2s_X - 2s_Y$, and the normalization factor relating symmetric states of different Hamming weight 
\begin{equation}
    \mathcal{N}(q_\lm, q_\lm') = \sqrt{\frac{q_\lm'! (n - 2m - q_\lm')! }{ q_\lm! (n - 2m - q_\lm)! }}\;. 
\end{equation}
Here, the phase factor is defined as 
\begin{equation}
    c(k_Y, \vec{a}, \vec{s}) = i^{k_Y}(-1)^{a_X + a_Z + s_Y + s_Z}\, \; 
\end{equation}
and the combinatorial weight counting all admissible Pauli assignments corresponds to
\begin{equation}
    W(q_\lm, \vec{k}, \vec{a}, \vec{s} ) = \frac{m!}{\vec{a}! (m - \norm{\vec{a}}_1)!}  \frac{q_\lm!}{\vec{s}!(q_\lm - \norm{\vec{s}}_1)!} \frac{(n - 2m - q_\lm)!}{(\vec{k} - 2\vec{a} - \vec{s})!( n - 2m - q_\lm - k + 2\norm{\vec{a}}_ 1 + \norm{\vec{s}}_1)!}\;, 
    \label{eq:combinatorial_factor}
\end{equation}
where we use the shorthand notation $\vec{v}! = v_X! v_Y! v_Z!$\; for a vector $\vec{v}$. The summations runs over integer-valued vectors $\vec{a} = (a_X, a_Y, a_Z)$ and $\vec{s} = (s_X, s_Y, s_Z)$ which are subject to componentwise inequalities, and   $\norm{\vec{a}}_1= a_X + a_Y + a_Z$ (similarly $ \norm{\vec{s}}_1= s_X + s_Y + s_Z$) denotes their $\ell_1$-norm. 
\end{theorem}

\begin{proof}
The proof exploits two facts: (i) the canonical Schur basis states consist of an antisymmetric register of $m$ singlet pairs and a fully symmetric register of $n-2m$ qubits; (ii) the symmetrized Pauli string $\TC_{S_n}(P_{\vec{k}})$ is permutation-invariant. Therefore, the action of $\TC_{S_n}(P_{\vec{k}})$ depends only on how many Pauli operators of each type act on the singlet pairs and on the $\ket{0}$/$\ket{1}$ positions in the symmetric register. We thus decompose the action of $\TC_{S_n}(P_{\vec{k}})$ into an antisymmetric contribution and a symmetric contribution, and count all admissible assignments combinatorially.

As stated in Eq.~\eqref{eq:schur_basis_expression} of the main text, the canonical Schur basis is defined as
\begin{equation}
\nonumber
\ket{\lambda,p_\lambda^0,q_\lambda}
=
\ket{\Psi}^{\otimes m}\otimes \ket{\Sigma_{q_\lambda}^{(n-2m)}},
\end{equation}
where $\ket{\Psi}=(\ket{01}-\ket{10})/\sqrt{2}$ is the two-qubit singlet state and $\ket{\Sigma_{q_\lm}^{(n-2m)}}$ is the normalized symmetric state of Hamming weight $q_\lm$ on $n-2m$ qubits. Since $\TC_{S_n}(P_{\vec{k}})$ is invariant under permutations, it suffices to count the number of ways Pauli operators can be assigned to these two sectors.

In what follows, we will decompose our proof into studying how an operator $\TC_{S_n}(P_{\vec{k}})$ acts over the Schur basis, first studying its action on the antisymmetric part, then on the symmetric part, and keeping track of phases, combinatorial factors and normalization.

We thus begin by considering the action of a symmetrized Pauli on the antisymmetric singlet pairs. For this purpose, consider a single singlet pair, $\ket{\Psi}$. As the operator $\TC_{S_n}(P_{\vec{k}})$ is fully symmetric while the state $\ket{\Psi}$ is antisymmetric, the action of all mixed Pauli pairs cancel and only identical Pauli contribute non-trivially. For instance, one can verify that the following cases hold
\begin{equation}
\TC_{S_2}(P\otimes Q)\ket{\Psi}
=
\begin{cases}
\ket{\Psi}, & P=Q=I,\\
-\ket{\Psi}, & P=Q\in\{X,Y,Z\},\\
0, & P\neq Q \text{ and } P,Q\in\{I,X,Y,Z\}.
\end{cases}
\label{eq:singlet_action}
\end{equation}

We introduce integers $a_\mu$ with $\mu \in \{X,Y,Z\}$, denoting the number of Pauli-$\mu$ pairs acting on the antisymmetric part, which satisfy
\begin{equation}
    a_\mu \in \left\{0, \dots, \left\lfloor \frac{k_\mu}{2}\right\rfloor \right\},~~~~~a_X + a_Y + a_Z = \norm{\vec{a}}_1 \le m\;. 
\end{equation} 
After allocating $2a_\mu$ Paulis to the singlet pairs, the remaining Pauli counts acting on the symmetric state are 
\begin{equation}
    k_\mu^{(s)} = k_\mu - 2a_\mu,\qquad \mu \in \{X,Y,Z\}\;, 
\end{equation}
where the superscript $(s)$ emphasizes the action on the symmetric state.

The next step is to study how symmetrized Paulis act on the symmetric register. Due to the permutation symmetry of the operators and the state, it suffices to evaluate the action of the symmetrized Pauli string solely on a representative computational basis 
\begin{equation}
 \ket{0}^{n - 2m - q_\lm} \otimes \ket{1}^{q_\lm}\;, \nonumber 
\end{equation}
 and then re-symmetrize. Let $\vec{s} = (s_X, s_Y, s_Z)$ denote how many of the $k_\mu^{(s)}$ operators act on the qubits in state $\ket{1}$ within this representative string. These integers satisfy
\begin{equation}
0 \le s_\mu \le k_\mu^{(s)},\qquad \norm{\vec{s}}_1 = s_X + s_Y + s_Z \le q_\lm\;. \nonumber
\end{equation}

The Pauli-$X$ and Pauli-$Y$ operators flip computational basis bits, while Pauli-$Z$ only contributes phases. Hence, the output Hamming weight is 
\begin{equation}
q'_\lm \equiv q'_\lm(\vec{k},\vec{a},\vec{s})
=
q_\lm
+\left(k_X^{(s)}+k_Y^{(s)}\right)
-2(s_X+s_Y)
=
q_\lm + k_X+k_Y -2a_X-2a_Y-2s_X-2s_Y,
\label{eq:qprime_def}
\end{equation}
which matches the theorem statement.

At this point we note the action of $\TC_{S_n}(P_{\vec{k}})$ will generate a phase one needs to keep track of. This phase comes from two sources: the singlet signs as discussed above, and phases acquired when $Y$ or $Z$ acts on the qubits in the symmetric state. Tracking these contributions yield the final phase factor stated in the theorem
\[
c(k_Y,\vec a,\vec s) = i^{k_Y}(-1)^{a_X+a_Z+s_Y+s_Z}\;,
\]

Then, it is essential to count how many Pauli assignments realize a given choice of $(\vec{a}, \vec{s})$. First, choosing the singlet allocations yields multinomial factor
\begin{equation}
    \frac{m!}{a_X! a_Y! a_Z! (m - \norm{\vec{a}}_1)!}\;. \nonumber
    \label{eq:combinatorial_factor1}
\end{equation} 
On the symmetric state, we choose which of the $q_\lm$ qubits initially in state $\ket{1}$ receive $s_X, s_Y, s_Z$ Pauli operators with the rest identity operators, giving
\begin{equation}
    \frac{q_\lm!}{s_X! s_Y! s_Z! (q_\lm - \norm{\vec{s}}_1)! }\;. \nonumber
    \label{eq:combinatorial_factor2}
\end{equation}
Finally, we assign the remaining $k_\mu^{(s)} - s_\mu = k_\mu - 2a_\mu - s_\mu$ Pauli-$\mu$ operators to the $(n - 2m - q_\lm) $ qubits in state$\ket{0}$, resulting in: 
\begin{equation}
    \frac{(n - 2m  - q_\lm)!}{(k_X - 2a_X - s_X)! (k_Y - 2a_Y - s_Y)! (k_Z - 2a_Z - s_Z)! (n - 2m - q_\lm - k + 2\norm{\vec{a}}_1 + \norm{\vec{s}}_1)!}\;. \nonumber
    \label{eq:combinatorial_factor3}
\end{equation}
Multiplying these factors exactly produces the combinatorial weight $W(q_\lm, \vec{k}, \vec{a}, \vec{s})$ defined in Eq.~\eqref{eq:combinatorial_factor}.  

To finish, the full action described above maps the representative weight-$q_\lm$ state to a representative weight-$q'_\lm$ state, and we need to re-normalize the mapping. Counting all the number of computational basis set in the normalized symmetric states and considering the normalization factors, we introduce the ratio of binomial normalization factors for $q_\lm$ and $q_\lm'$ states
\[
\NC(q_\lm,q'_\lm)
=
\sqrt{\frac{q'_\lm!\,(n-2m-q'_\lm)!}{q_\lm!\,(n-2m-q_\lm)!}},
\]
which matches the theorem.

Combining all the combinatorial coefficients and phase factors above with the symmetrization prefactor $\frac{k_X! k_Y! k_Z! (n-k)!}{n!}$, and summing over all admissible $\vec{a}$ and $\vec{s}$ satisfying the stated constraints yield Eq.~\eqref{eq:action_schur}. 
\end{proof}

\begin{algorithm}[h]

\caption{Matrix elements of $\TC_{S_n}(P_{\vec{k}})$ in the canonical Schur basis}
\label{alg:schur_pauli_matrix_elements}
\KwIn{$n$ (number of qubits); Pauli weights $\vec{k}=(k_X,k_Y,k_Z)$ with $k=k_X+k_Y+k_Z$; Irrep label $\lm \equiv \lm(m) = (n -m, m)$; Basis label $q_\lm$.}
\KwOut{Nonzero matrix elements $\mel{\lambda,p_\lambda^0, q_\lm}{\TC_{S_n}(P_{\vec{k}})}{\lambda,p_\lambda^0,  q'_\lm}$ for all reachable $q'_\lambda$.}

\BlankLine
\textbf{Initialize} an associative array (or sparse map) $M[q'] \leftarrow 0$ for $q'=0,1,\dots,n-2m$\;

\BlankLine
\For(\tcp*[f]{Enumerate Pauli pairs on the antisymmetric component}){$a_X=0$ \KwTo $\min(m, \left\lfloor k_X/2\right\rfloor)$}{
  \For{$a_Y=0$ \KwTo $\min(m - a_X, \left\lfloor k_Y/2\right\rfloor)$}{
    \For{$a_Z=0$ \KwTo $\min(m - a_X - a_Y, \left\lfloor k_Z/2\right\rfloor)$}{
      $\vec{a}\leftarrow(a_X,a_Y,a_Z)$;
      
      \BlankLine
      \tcp{Remaining Pauli operators acting on the symmetric component}
      $k_X^{s}\leftarrow k_X-2a_X$\;
      $k_Y^{s}\leftarrow k_Y-2a_Y$\;
      $k_Z^{s}\leftarrow k_Z-2a_Z$\;

      \BlankLine
      \For(\tcp*[f]{Enumerate the number of Paulis hitting the $\ket{1}$-support}){$s_X=0$ \KwTo $\min(q_\lambda,k_X^{s})$}{
        \For{$s_Y=0$ \KwTo $\min(q_\lambda-s_X,k_Y^{s})$}{
          \For{$s_Z=0$ \KwTo $\min(q_\lambda-s_X-s_Y,k_Z^{s})$}{
            $\vec{s}\leftarrow(s_X,s_Y,s_Z)$;\quad $S\leftarrow \|\vec{s}\|_1$\;

            \BlankLine
            \tcp{Target Hamming weight after action}
            $q' \leftarrow q_\lambda + k_X + k_Y -2a_X-2a_Y-2s_X-2s_Y$\;
            \If{$q'<0$ \textbf{or} $q'>n-2m$}{\textbf{continue}}
            
            \BlankLine
            \tcp{Phase factor}
            $c \leftarrow i^{k_Y}(-1)^{a_X+a_Z+s_Y+s_Z}$\;

            \BlankLine
            \tcp{Combinatorial weight $W(q_\lambda,\vec{k},\vec{a},\vec{s})$}
            $W \leftarrow
            \dfrac{m!}{\vec{a}!\,(m-A)!}
            \cdot
            \dfrac{q_\lambda!}{\vec{s}!\,(q_\lambda-S)!}
            \cdot
            \dfrac{(n-2m-q_\lambda)!}{
              (\vec{k}-2\vec{a}-\vec{s})!\,
              \bigl(n-2m-q_\lambda-\|\vec{k}\|_1+2\norm{\vec{a}}_1+ \norm{\vec{s}}_1\bigr)!
            }$\;

            \BlankLine
            \tcp{Accumulate the matrix element contribution}
            $M[q'] \leftarrow M[q'] + c\cdot W$\;
          }
        }
      }
    }
  }
}

\For(\tcp*[f]{Multiply by the normalization factor}){$q' = 0$ \KwTo $n - 2m$} {
    $M[q'] \leftarrow \frac{k_X! k_Y! k_Z! (n-k)!}{n!} \cdot \sqrt{\frac{q_\lm! (n - 2m  - q_\lm)!}{q'! (n - 2m  - q')!}}\cdot M[q']$
}
\BlankLine
\Return{$\{(q',\,M[q'])\}_{q'=0}^{n-2m}$ as the nonzero matrix elements 
$\mel{\lambda,p_\lambda^0, q'}{\TC_{S_n}(P_{\vec{k}})}{\lambda,p_\lambda^0, q_\lm}$}\;
\end{algorithm}

The previous steps are summarized in Algorithm~\ref{alg:schur_pauli_matrix_elements}, which provides an explicit and exhaustive procedure for evaluating the matrix elements of the symmetrized Pauli string $\TC_{S_n}(P_{\vec{k}})$ in the canonical Schur basis. The procedure is efficient when the Pauli weight is constant, i.e., $k \in \order{1}$, but becomes computationally impractical in the regime $k \in \order{n}$. Based on this algorithm, we now analyze the computational complexity in the regime $k \in \order{1}$.  

\begin{corollary}[Computational complexity for the matrix element computation of symmetrized Pauli strings in the Schur basis]
\label{cor:cost_matrix_calculation}
If the Pauli weight $k=k_X+k_Y+k_Z$ is constant, i.e.\ $k\in \mathcal{O}(1)$,
then all non-zero matrix elements of $\TC_{S_n}(P_{\vec{k}})$ in the canonical Schur basis, defined by Eq.~\eqref{eq:action_schur}, can be evaluated with total computational cost $\mathcal{O}(n^2)$.
\end{corollary}

\begin{proof}
For simplicity, we assume throughout the proof that matrix initialization and factorial evaluations appearing within a single loop iteration take $\order{1}$ time. 
Algorithm~\ref{alg:schur_pauli_matrix_elements} evaluates a matrix element of a symmetrized Pauli string by iterating over integer vectors $\vec{a}$ and $\vec{s}$ subject to the componentwise constraints
\begin{equation}
    \vec{a}\le \left\lfloor \frac{\vec{k}}{2}\right\rfloor,
\qquad
\vec{s}\le \vec{k}-2\vec{a}, \nonumber
\end{equation}
together with the $\ell_1$-norm conditions $\norm{\vec{a}}_1 \le m$ and $\norm{\vec{s}}_1 \le q_\lm$. 
When $k \in \order{1}$, the total number of admissible combinations of $\vec{a}$ and $\vec{s}$ is bounded by a constant that depends only on $\vec{k}$ and is independent of $n$. Therefore, the cost of evaluating all contributions for fixed $(m, q_\lm)$ scales as $\order{(k_X\cdot k_Y \cdot k_Z)^2}\sim\order{1}$. 

Moreover, to obtain the full action for a fixed irrep label $\lm =(n-m, m)$, one must evaluate Eq.~\eqref{eq:action_schur} for all the Hamming weight $q_\lm = 0,\dots,n-2m$, which requires $\Theta(n - 2m)$ calls to Algorithm~\ref{alg:schur_pauli_matrix_elements}. Since each call costs $\order{1}$, the total time complexity for fixed $m$ is $\order{n - 2m}$. 
Finally, summing over all possible values $m = 0, 1,\dots, \left\lfloor \frac{n}{2}\right\rfloor$ yields a total runtime of
\[
\sum_{m=0}^{\lfloor n/2\rfloor} \mathcal{O}(n-2m)\;\sim\; \mathcal{O}(n^2).
\]
Therefore, for constant Pauli string weight $k \in \order{1}$, all nonzero matrix elements of $\TC_{S_n}(P_{\vec{k}})$ in the canonical Schur basis can be computed in total time $\order{n^2}$. 
\end{proof}

In practice, the assumption of constant-time factorial evaluation is satisfied by precomputing the corresponding log-factorials $\log(t!)$ for $t = 1,\dots,n$ once at the start of the algorithm, thereby, avoiding repeated factorial evaluations within the inner loops.
To validate the claim of Corollary~\ref{cor:cost_matrix_calculation}, Figure~\ref{fig:simulation_time_matrix} shows the numerical simulation time required to compute the matrix elements for all irrep blocks $\lambda$ as a function of the number of qubits. 

\begin{figure}
\subfloat[$k_X = k_Y = k_Z = 1$]{\includegraphics[width=0.33\linewidth]{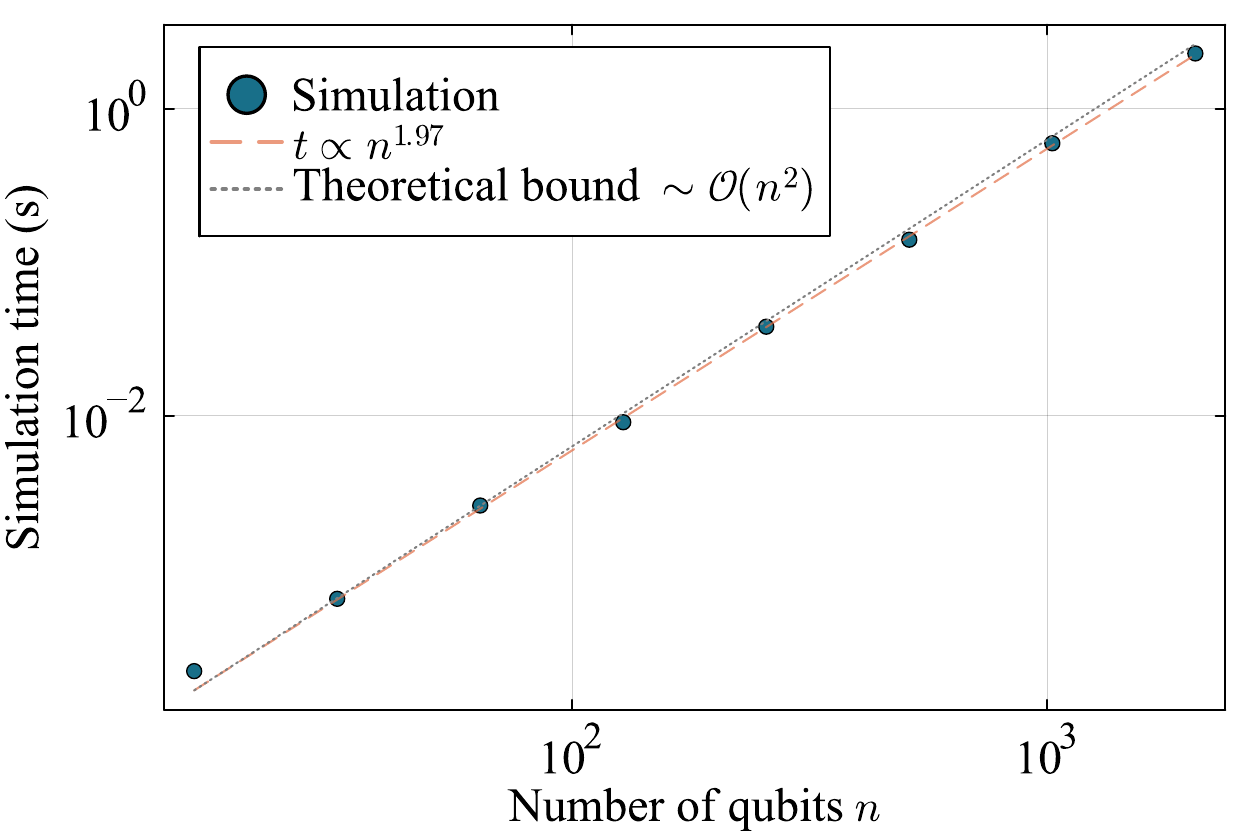}}
\subfloat[$k_X = k_Y = k_Z = 5$]{\includegraphics[width=0.33\linewidth]{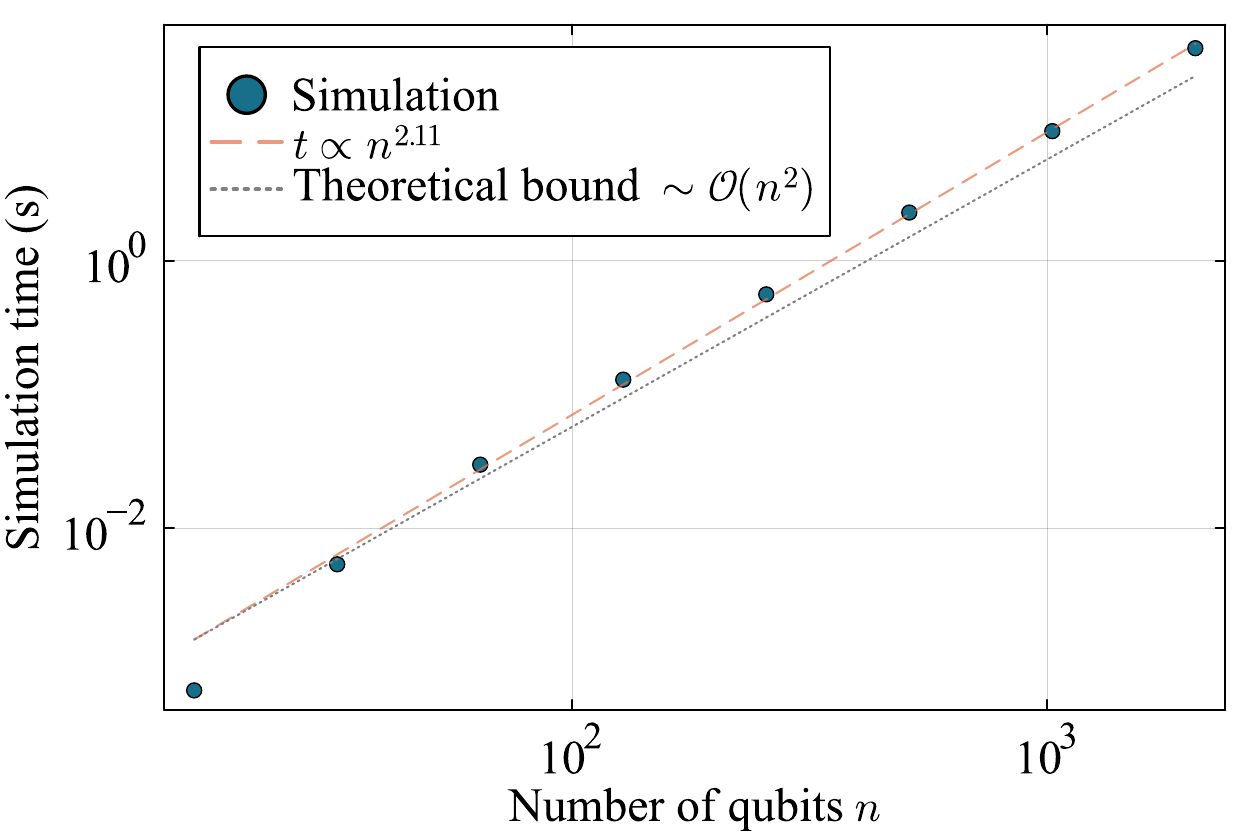}}
\subfloat[$k_X = 4, k_Y = 2, k_Z = 6$]{\includegraphics[width=0.33\linewidth]{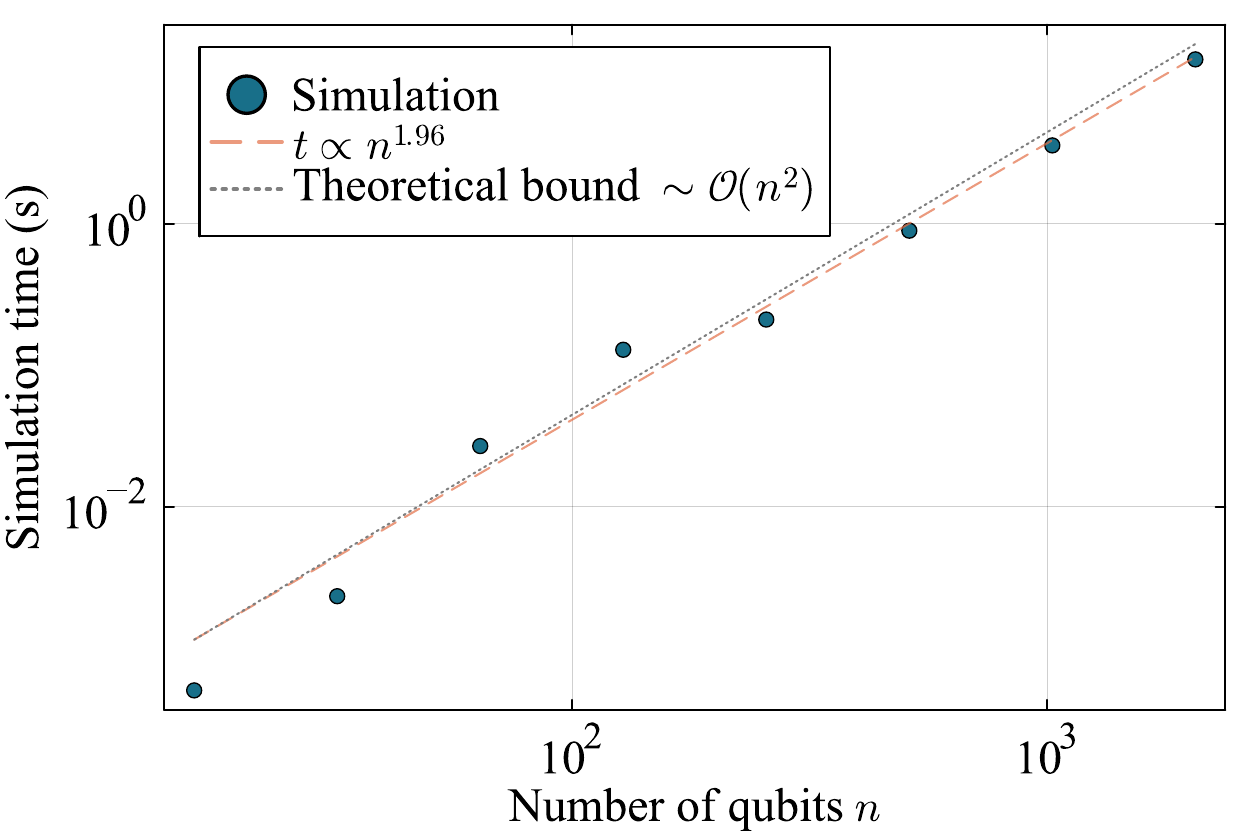}}
\caption{\textbf{Simulation time required to compute the matrix elements of $\TC_{S_n}(P_{\vec{k}})$ in the Schur basis for different choices of $\vec{k}$ with $k \in \order{1}$ using Algorithm~\ref{alg:schur_pauli_matrix_elements}.} All the simulations were performed using the Julia package and factorial evaluations inside the algorithm were optimized by precomputing log-factorials at the start. The observed simulation time scales at most as $\order{n^2}$, in agreement with the claim of Corollary~\ref{cor:cost_matrix_calculation}, with a slightly higher runtime observed in (b). This discrepancy is due to additional practical overhead, primarily arising from the evaluation of bound conditions. }
    \label{fig:simulation_time_matrix}
\end{figure}

\section{Classical shadow}
Before presenting the details of quantum sample complexity for $S_n$-equivariant quantum models, we first provide a brief overview of the classical shadow protocol~\cite{huang2020predicting, sack2022avoiding, zhao2021fermionic, bertoni2024shallow} which estimates the expectation values $\Tr[\rho O_i]$ for a set of $M$ observables $\{O_i\}_{i = 1}^M$ given a $d$-dimensional quantum state. 
The classical shadow protocol extracts information from an unknown quantum state $\rho$ by evolving it under a random unitary $V$ sampled from an ensemble of unitaries, $\VC$, and then measuring it in a fixed orthonormal basis $\WC = \{\ket{w_i}\}_{i = 1}^d$, satisfying $\sum_{i = 1}^d \Pi_{w_i} = \eye$ with the projectors $\textbf{P} = \{\Pi_{w_i} = \ketbraq{w_i}\}_{w_i}$. The output of a shadow, namely the tuple $\{V,\ket{w}\}$ is called a snapshot, and we denote as $\MC$ the quantum channel that maps the state into the snapshot. 

More specifically,  $\MC$ is given by
\begin{align}
    \MC_{\VC}(\rho) & = \Ebb_{V \thicksim \VC, w, \rho}\left[V\ad  \Pi_w V \right] = \sum_w \int_{V\thicksim \VC} dV \, \mu(V) \Tr[\rho V\ad \Pi_w V] V\ad  \Pi_w V \nonumber \\  
    & = \Tr_{A}\left[\left(\Ebb_V\left[\left(V\ad\right)^{\otimes 2}\tilde{\Pi}V^{\otimes 2}\right]\right)\Big(\rho \otimes \eye\Big)\right],~~~\text{with } \tilde{\Pi} = \sum_w \left(\Pi_w\right)^{\otimes 2} 
        \;.
    \label{eq:measurement_channel}
\end{align}
where the expectation value $\Ebb_{V \thicksim \VC, w, \rho}$ is taken with respect to the probability distribution $p_\rho(V, w) = \mu(V)\Tr[V\ad \rho V \Pi_w]$ given the sampling probability $\mu(V)$ for $V\sim \VC$. In the last line, the measurement channel is expressed on $2n$-qubits, where we trace out the first $n$-qubits via $\Tr_A[\; \cdot\;]$. 

Our goal is to recover an unbiased estimator of $\rho$ by inverting the action of $\MC$ on the snapshots. That is, we need to compute  the pseudo-inverse $\MC_{\VC}^{-1}$\footnote{If the unitary ensemble $\VC$ is tomographically complete, then the corresponding measurement channel $\MC$ is invertible, allowing one to explicitly construct its inverse. } 
\begin{equation}
    \rholh_{V, w} = \MC^{-1}_{\VC}\left(V\ad \Pi_w V\right)\;.
\end{equation}

By repeating this procedure $N$ times, one collects $N$ independent snapshots with the corresponding measurement data $\{V,\ket{w}\}$. 
Since $\rholh_{V, w}$ forms unbiased estimators of $\rho$, we can compute an unbiased estimator $\hat{o}_i$ for the expectation value of the observable $O_i$ given by  
\begin{equation}
    \hat{o}_i= \Tr[\rholh_{V, w} O_i ] = \Tr[O_i \MC^{-1}(V\ad \Pi_{w} V)]\;,  
    \label{eq:cs_estimates}
\end{equation}
satisfying $\Ebb_{V\sim \VC, w}[\hat{o}_i]  = \Tr[\rho O_i]$.

In particular, the quantum channel $\MC$ defines an operator subspace $\LC_{\rm vis}$, so-called visible operator space~\cite{van2022hardware}, consisting of all the operators that can be estimated via the classical shadow: 
\begin{equation}
    \LC_{\rm vis} = \frak{Im}(\MC) =  \spn\{V\ad \Pi_w V\; \lvert \; V \in \VC,~\ket{w}\in \WC \}\; \subseteq \LC(\HC)\;, 
    \label{eq:visible_space}
\end{equation}
where $\frak{Im}(\MC)$ denotes the image of $\MC$ and $\LC(\HC)$ is the space of linear operators acting on the Hilbert space $\HC$. 

A central goal of the classical shadow is to minimize the number of snapshots $N$ required to accurately estimate the expectation values. This is controlled by the variance of the estimator, 
\begin{equation}
    \Var[\hat{o}_i]  = \Ebb_{V\sim \VC, w}[\hat{o}_i^2] - \left(\Ebb_{V\sim \VC, w}[\hat{o}_i] \right)^2 \le \norm{O_i}^2_{\rm shadow} 
    \end{equation}
where $\norm{O_i}^2_{\rm shadow}$ denotes the shadow norm defined as  
\begin{equation}
    \norm{O_i}^2_{\rm shadow} = \max_{\sigma : {\rm state}} \Ebb_{V\thicksim\VC,~w}\left[ \Tr[\MC^{-1}(V\ad \Pi_w V)O ]^2\right]\;.
    \label{eq:shadow_norm}
\end{equation}
The previous equation showcases that the shadow norm explicitly depends on the unitary ensemble $\VC$ as well as the measurement basis $\WC$~\cite{sauvage2024classical, wan2022matchgate, holtz2012alternating, west2024real}.

This leads to an upper bound on the number of snapshots $N$ required for the classical shadow protocol to accurately estimate the set of expectation values $\{\Tr[\rho O_i]\}_{i=1}^M$ up to an (additive) error $\ep$ and the success probability of at least $1- \dl \in [0, 1]$ as follows
    \begin{equation}
        N \in \order{\frac{\log(M/\dl)}{\ep^2} \max_{i = 1, \dots M} \norm{O_i}^2_{\rm shadow}}\;,
    \label{eq:classical_shadow}
    \end{equation}

\section{Deep permutation-invariant classical shadow\label{adx:pi_cs}}
In this section, we provide a review of a variant of classical shadows, tailored for quantum circuits with symmetry, specifically permutation group. We focus on Permutation Invariant Classical Shadow (PI-CS)~\cite{sauvage2024classical}, which is applied to a block-diagonalized quantum state by preparing the initial state $\rho$ and implementing the Schur transform on a quantum computer. 

\subsection{Symmetric classical shadow}

Given a group $G$ with unitary representation $R \to \Ubb(\HC)$, the goal of symmetric classical shadow protocol is to choose the measurement primitives $\VC$ and the set of projections $\mathbf{P}$ such that the corresponding visible operator space satisfies
\begin{equation}
    \LC_{\rm vis} = \LC^G = \left\{A \in \LC(\HC) \; \middle \lvert \; [A, R(g)] = 0,~\forall g \in G \right\} \;,  
\end{equation}
that is, all operators in the visible space (c.f. Eq.~\eqref{eq:visible_space}) should be equivariant with respect to $G$. Here, we denote $\LC^G = \text{comm}(G)$ to explicitly express that it is a space of linear operators.

\begin{figure}[h]
    \centering
    \includegraphics[width=0.55\textwidth, trim = {0.5cm, 0.5cm, 0.5cm, 0.5cm}]{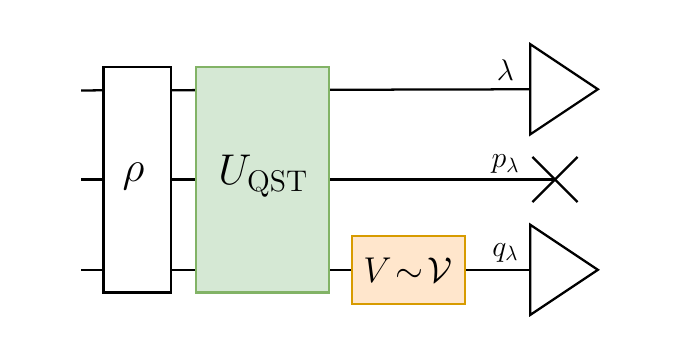}
    \caption{\textbf{Schematic diagram of Permutation Invariant Classical Shadow.} The initial state $\rho$ is transformed into a block-diagonal form by the quantum circuit responsible for the Schur transform $\QST$. We perform the permutation invariant classical shadow by measuring only the irrep register $\lm$ and the dimension register $q_\lm$, while leaving the multiplicity register $p_\lm$ untouched.}
    \label{fig:pi_cs}
\end{figure}

Let us formalize the symmetric classical shadow framework for estimating the expectation value $\Tr[\rho O_i]$ of a quantum state $\rho$, where the observables $\{O_i\}_{i = 1,\dots,M}$ are equivariant with respect to a symmetry group $G$. By complete reducibility of finite-dimensional unitary representations, there exists a basis set $\{\ket{\lm, p_\lm, q_\lm}\}_{\lm, p_\lm, q_\lm}$ in which the unitary representation $R$ admits an isotypic decomposition 
\begin{equation}
    R(g) \cong \bigoplus_\lm r_\lm \otimes \eye_{d_\lm}\;, 
\end{equation}
where $r_\lm$ is the irrep of dimension $m_\lm := \text{dim}(r_\lm)$ and $d_\lm := \text{mult}(r_\lm)$ denotes its multiplicity.
Accordingly, the Hilbert space admits the corresponding isotypic decomposition
\begin{equation}
    \HC \cong \bigoplus_\lm \HC_\lm \cong \bigoplus_\lm \left(S_\lm \otimes W_\lm \right)\;, 
\end{equation} 
where $S_\lm$ is the space of the irrep $r_\lm$ and $W_\lm \cong \text{Hom}_G(S_\lm, \HC)$ is the multiplicity space, i.e., the space of  $G$-equivariant linear maps from $S_\lm$ to $\HC$ whose dimension equals the multiplicity $d_\lm$.
Under the same basis change, the $G$-equivariant operator $A$ admits a block-diagonal decomposition of the form 
\begin{equation}
    A \cong  \bigoplus_\lm \eye_{m_\lm} \otimes A_\lm\;,  \label{eq:irrep_A_adx}
\end{equation}
where $A_\lm$ acts only on the multiplicity space $W_\lm$, and therefore has dimension $d_\lm$.

We assume that the isotypic basis can be decomposed into a tensor product of separate registers as $\ket{\lm, p_\lm, q_\lm} = \ket{\lm} \otimes \ket{p_\lm} \otimes \ket{q_\lm}$. We will refer to each register as irrep, multiplicity and dimension registers, respectively, with $p_\lm  \in [m_\lm]$ and $q_\lm\in [d_\lm]$. Throughout, the terminology ``multiplicity'' and ``dimension'' refer to the block structure of $G$-equivariant operators (i.e., the commutant of the group action), rather than to the decomposition of the group representation itself, in order to remain consistent with the main text.
In particular, storing a basis label $q_\lm$ requires $\lceil \log_2(d_{\max})\rceil$ qubits, where $d_{\max} = \max_\lm (d_\lm)$ denotes the largest irrep dimension of the group-equivariant operators. For example, in the case of the permutation group, the largest irrep dimension of the $S_n$-equivariant operators scales as $d_{\max} = \max_\lm(d_\lm) = n+1$. Hence, the dimension register consists of $\lceil \log_2(n+1) \rceil$ qubits, while the associated Hilbert space has dimension $n+1$. Finally, we emphasize that the definitions of dimension and multiplicity differ from those used in the original study~\cite{sauvage2024classical}. In this work, we define these notions for the irrep blocks of a commutant algebra whereas Ref.~\cite{sauvage2024classical} defines them for irreducible components in unitary representations of the group.  These roles are effectively interchanged between the two contexts, which can lead to confusion if not carefully distinguished. In particular, what we denote as the largest dimension of the group-commutant irrep, $d_{\max}$, corresponds to the largest multiplicity of a group irrep underlying the shadow protocol.

The key idea of the symmetric classical shadow is to perform the measurement only on the irrep and dimension registers, while leaving the multiplicity register untouched. Accordingly, we consider random unitaries $V \in \VC$ that has the form $V = \eye \otimes \eye \otimes V_\lm$, where $V_\lm$ acts solely on the dimension register, as illustrated in Fig.~\ref{fig:pi_cs}. 

The symmetric shadow protocol proceeds as follows:
\begin{enumerate}
    \item Apply a quantum circuit $U_{\rm block}$ to block-diagonalize the initial state. 
    \item Measure the irrep register $\ket{\lm}$, obtaining the irrep label $\lm$.
    \item Apply a random unitary $V_\lm \in \VC$ to the dimension register $\ket{q_\lm}$. 
    \item Measure the dimension register, yielding outcome $q_\lm$.   
\end{enumerate}
Within this procedure, the corresponding projective measurement operators are given by:
\begin{equation}
    \mathbf{P}^G  = \left\{\Pi_{\lm, q_\lm}\right\}_{\lm, q_\lm} = \left\{ \sum_{p_\lm=1}^{m_\lm} \ketbraq{\lm}\otimes \ketbraq{p_\lm}  \otimes \ketbraq{q_\lm} \right\}_{q_\lm = 1} ^{d_\lm}\;,
\end{equation}
and given the unitary primitives $\VC$, the resulting visible space $\LC_{\rm vis}$ contains operators:
\begin{equation}
    V\ad \Pi_{\lm, q_\lm} V = \sum_{p_\lm = 1}^{m_\lm} \ketbraq{\lm}  \otimes \ketbraq{p_\lm} \otimes V_\lm\ad \ketbraq{q_\lm}V_\lm = \ketbraq{\lm} \otimes \eye \otimes V_\lm\ad \ketbraq{q_\lm}V_\lm \,.
\end{equation}
Therefore, $\LC_{\rm vis} \subseteq  {\rm comm}(G)$, which consists of $G$-equivariant observables of the form in Eq.~\eqref{eq:irrep_A_adx}. Thus, it suffices to apply standard classical shadow techniques---such as local or global Clifford measurements---on $q$-qubits of the dimension register.

Mathematically, the expectation value can be written as a sum over the irreps: 
\begin{equation}
    o = \Tr[\rho O] = \sum_\lm   \Tr[\rho \left(\eye^{m_\lm} \otimes O_\lm \right)]  = \sum_\lm \sum_{p_\lm=1}^{m_\lm} \Tr[\rho^{p_\lm}_{\lm} O_\lm]\;. 
    \label{eq:exp_val_symm_CS}
\end{equation}
 Here $\rho_\lm^{p_\lm}$ denotes the $p_\lm$-th diagonal multiplicity block, defined as the restriction of the state $\rho$ to the irrep subspace $\HC_\lm^{p_\lm} = \spn\{\ket{\lm, p_\lm, q_\lm}\}_{q_\lm=1}^{d_\lm}$. More precisely, $\rho_\lm^{p_\lm}$ is an operator acting on $\HC_\lm^{p_\lm}$, i.e, $\rho_\lm^{p_\lm} \in \BC(\HC_{\lm}^{p_\lm}) \cong \BC(\Cbb^{d_\lm})$, and is obtained by projecting $\rho$ onto the $p_\lm$-th copy of the irrep as
\begin{equation}
    \rho_\lm^{p_\lm} =  (\bra{\lambda, p_\lm}\otimes \eye )\, \rho\, (\ket{\lm, p_\lm} \otimes \eye)\;, 
\end{equation}
where we denote $\ket{\lambda, p_\lm}$ the orthonormal basis living in the multiplicity space, $\Cbb^{m_\lm}$. 

It is worth noting that, in Ref.~\cite{sauvage2024classical}, the authors restrict attention to symmetric states of the form $\rho \cong \bigoplus_\lm \eye_{m_\lm}\otimes \rho_\lm $, for which the reduced state $\rho_\lm$ is identical across all multiplicity sectors. In this case, the sum in Eq.~\eqref{eq:exp_val_symm_CS} simplifies to $o = \sum_\lm m_\lm \Tr[\rho_\lm O_\lm]$ as the state $\rho_\lm$ is repeated $m_\lm$ times. However,  for a general input state $\rho$, the reduced states $\rho^{p_\lm}_\lm$ corresponding to different multiplicity labels $p_\lm$ need not coincide, as visually explained in Fig.~\ref{fig:rho_symmetric}.

Upon measuring the irrep label, the outcome label $\lm$ is obtained with probability $p(\lm) = \Tr[\tau_\lm]$  where we denote $o_\lm = \Tr[\tau_\lm O_\lm]$ with $\tau_\lm = \sum_{p_\lm=1}^{m_\lm} \rho_\lm^{p_\lm}$. If $\rho$ is also $G$-equivariant, the probability simplifies to $p(\lm) = m_\lm \Tr[\rho_\lm ]$, which is presented in Ref.~\cite{sauvage2024classical}.   Tracing out the multiplicity space yields the normalized reduced state $\tilde{\tau}_\lm = \tau_\lm /\Tr[\tau_\lm]$. Applying the classical shadow protocol on the dimension register then allows one to construct an unbiased estimator $\hat{o}_\lm$ of $o_\lm = \Tr[\tilde{\tau_\lm} O_\lm]$ such that the total expectation value is recovered as $o = \sum_\lm p(\lm) o_\lm$. 

It follows that the unbiased estimator $\hat{o}$ of $o$ has variance defined as 
\begin{equation}
    \Var[\hat{o}]_{\rm{block-CS}} = \left(\sum_\lm p(\lm)\Var[\hat{o}_\lm]_{\rm{CS}}\right) + \Var_\lm o_\lm\;.
\end{equation}
Here, we use the subscripts  $\textrm{block-CS}$ and $\textrm{CS}$ to denote that the variance arises from the symmetric classical shadow and the standard classical shadow applied within each irrep sector, respectively. The first term corresponds to the average shadow-induced variance within each sector $\lm$, while the second term captures the classical variance associated with sampling over irrep labels. Using the fact that $\Var[\hat{o}_\lm]_{\rm CS} \le \norm{O_\lm}_{\rm CS}\le d_{\max}^2\norm{O}_\infty^2$ and $\Var_\lm o_\lm \le \norm{O_\lm}^2_\infty \le \norm{O}^2_\infty$ together with standard  shadow norm bounds, we obtain the following variance bounds: 
\begin{equation}
    \Var[\hat{o}]_{\rm{block-LC}} \le (d_{\max}^2 + 1) \norm{O}^2_{\infty} 
    \label{eq:var_block_lc}
\end{equation}
\begin{equation}
    \Var[\hat{o}]_{\rm{block-GC}} \le 3(d_{\max}^2 + 1) \norm{O}^2_{\infty} 
    \label{eq:var_block_gc}
\end{equation}
with the notation $\text{LC}$ for local Clifford and $\text{GC}$ for global Clifford measurements on the dimension register. The result implies that the sample complexity of the symmetric shadow protocol is bounded by the largest irrep dimension of the group-equivariant operators, which is equivalently given by the largest irrep multiplicity of the underlying group representation. 

\subsection{Permutation-invariant classical shadow}
We now focus on the case $G = S_n$ and formalize the PI-CS framework. The most direct implementation uses quantum Schur transform (QST), denoted as $\QST$,  which changes the basis from the computational basis into Schur basis. 
Efficient quantum circuits implementing QST have already been proposed with the quantum time complexity of $\order{n \cdot \poly(\log \ep_{\rm QST}^{-1})}$ where $\ep_{\rm QST}$ quantifies the accumulated approximation error~\cite{bacon2005quantum, bacon2006efficient, kirby2017practical, anschuetz2022efficient}.

Under this setting, following Eqs.~\eqref{eq:var_block_lc} and \eqref{eq:var_block_gc}, the variance of the estimated expectation values can be bounded as~\cite{sauvage2024classical} 
\begin{equation}
    \Var[\hat{o}]_{\rm QST-LC}\le (n^2 + 2n + 2)  \norm{O}_\infty^2,  
\end{equation}
and
\begin{equation}
    \Var[\hat{o}]_{\rm QST-GC}\le 3(n^2 + 2n + 2)  \norm{O}_\infty^2\,,  
\end{equation}
where we use $d_{\max} = n+1$ for the Schur basis. 
In particular, for the (normalized) symmetrized Pauli operators defined in Eq.~\eqref{eq:Sn_observables}, the spectral norm satisfies $\norm{O}_{\infty}\le 1$ by its construction.

Following the bounds of the variance, we obtain the following theorem for quantum sample complexity using both local and global Clifford measurements.
\begin{theorem}[Quantum Sample Complexity of deep PI-CS]
    Consider a general input state $\rho$ and a  collection of $M$ $S_n$-equivariant observables $O_i$, as well as parameters $\ep, \dl > 0$. Then, with probability (at least) $1-\dl$, classical shadows of size 
    \begin{equation}
        N \in \order{\log(\frac{M}{\dl})\frac{n^2}{\ep^2} \max_i \norm{O_i}^2_\infty} \;, 
    \end{equation}
    suffice to jointly estimate all $M$ expectation values up to additive accuracy $\ep$.
\label{thm:classical_deep_shadow_PI}
\end{theorem}

Although this approach is conceptually straightforward, it requires substantial circuit overhead. In particular, implementing QST requires polynomial depth circuit and additional qubits, since ancillary registers are needed. Another alternative method that mitigates this quantum overhead is shallow permutation classical shadows which achieves quantum data acquisition using $\order{1}$ circuit depth by avoiding the explicit implementation of the Schur transform. 
However, this advantage comes at the cost of increased classical post processing. In particular,  reconstructing components of the state projected onto irrep subspace becomes computationally expensive, while access to the quantity $\sum_{p_\lm = 1}^{d_\lm} \rho_{\lm}^{p_\lm}$ is essential for efficient classical simulation as highlighted in Eq.~\eqref{eq:loss_block}. This underlines the inherent trade-off between quantum resources (circuit depth and qubit count) and classical resources (post-processing). 

\section{Efficient PI-CS shadow for $S_n$-equivariant input state\label{adx:symmetrized_pi_cs}}
Now, we consider the case where the initial state is also $S_n$-equivariant, i.e., $\rho = \bigoplus_\lm \eye_{m_\lm} \otimes \rho_\lm$ with $\rho_\lm$ a $d_\lm \times d_\lm$ sized matrix. In this case, we are allowed to use symmetrized PI-CS, which avoids implementing a QST. In this section, we summarize the symmetrized PI-CS protocol formulated entirely in Pauli basis and Schur basis. Reader may refer to Appendix B and C in Ref.~\cite{sauvage2024classical} for further details as well as rigorous mathematical proofs of the method.

\subsection{Symmetrized PI-CS protocol}

In symmetrized PI-CS, we take the projective measurement ensemble $\VC_{\rm symm-PI}$ defined as: 
\begin{equation}
    \VC_{\rm symm- PI} = \{V = W^{\otimes n} \; \lvert \; W\in \mathbb{SU}(2) \} 
\end{equation}
In practice, each run draws a single-qubit gate $W$ from the Haar measure on $\mathbb{SU}(2)$ by parameterizing $W$ with a vector of Euler-angles $\vec{\bt} = [\bt_1, \bt_2, \bt_3]$ as
\begin{equation}
    W_{\vec{\bt}} = e^{-i \frac{\bt_3}{2} Z} e^{-i \frac{\bt_2}{2} Y} e^{-i \frac{\bt_1}{2} Z},~~~\text{with } \vec{\bt}~\text{sampled from } 
    \begin{cases}
        p(\bt_1) = p(\bt_3) = \frac{1}{2\pi},  & \bt_1, \bt_3 \in [0, 2\pi]\;, \\ 
        p(\bt_2) = \frac{\sin(\bt_2)}{2}, & \bt_2 \in [0,\pi]\;.
    \end{cases}
\label{eq:W_theta}
\end{equation}

After applying $V = W^{\otimes n}$, the measurement is performed with a set of Hamming weight projectors, which consists of $n+1$ projectors given by
\begin{equation}
    \Pi_{h} = \sum_{\substack{x \in \{0, 1\}^n \\ {\rm HW}(x) = h}} \ketbraq{x} = \frac{1}{h!(n-h)!} \sum_{\sg \in S_n} \sg\left(\ketbraq{0}^{\otimes(n-h)} \ketbraq{1}^{\otimes h } \right)\sigma\ad\;. 
    \label{eq:HW_projectors}
\end{equation}
Then, the classical-shadow estimator for an observable $O$ can be explicitly written as:
\begin{equation}
    \hat{o} = \Tr[\hat{\rho} O] = \Tr[O \MC^{-1}_{\rm symm-PI} \left((W\ad_{\vec{\bt}})^{\otimes n } \Pi_h W_{\vec{\bt}}^{\otimes n}\right)]\;.
    \label{eq:o_estimate}
\end{equation}

\subsection{Constructing classical shadow estimates}
To construct the classical shadow estimates for symmetrized PI-CS, we evaluate the expectation value $\Ebb_{V}[(V\ad)^{\otimes 2} \tilde{\Pi} V^{\otimes 2}] = \Ebb_{\vec{\bt}}[(W_{\vec{\bt}}\ad)^{\otimes 2n } \tilde{\Pi} W_{\vec{\bt}}^{\otimes 2n}] $ in Eq.~\eqref{eq:measurement_channel}. 
We introduce an orthonormal symmetrized Pauli basis set which consists of all the symmetrized Pauli string $\TC_{S_n}(P_{\vec{k}})$ defined over $n$-qubits 
\begin{equation}
    B_{\rm Pauli} = \left\{ B_{\vec{k}} = \TC_{S_n}\left(P_{\vec{k}}\right) = \frac{1}{\sqrt{d n! (n-k)! \vec{k}!} }  \sum_{\sg \in S_n} R(\sg)\left(X^{\otimes k_X} Y^{\otimes k_Y } Z^{\otimes k_Z} \eye^{\otimes n -k}\right)R\ad(\sg) \right\}_{\vec{k}}\;, 
    \label{eq:B_pauli}
\end{equation}
where  $d = 2^n$ denotes the Hilbert space dimension.
Here, we use a shorthanded notation $\vec{k}! = k_X! k_Y! k_Z!$ and $P_{\vec{k}}$ is defined as Eq.~\eqref{eq:Pk} with $k = \abs{\vec{k}}$. The basis spans $\mathrm{comm}(S_n)$ with dimension 
\begin{equation}
    d_{\rm PI} = |B_{\rm Pauli}|= \binom{n+1}{3} = \frac{(n+1)(n+2)(n+3)}{6}\;.
\end{equation}  Expressed in this basis, the two-copy moment admits the expansion
\begin{equation}
    \Ebb_{W}\left[ (W_{\vec{\bt}}\ad)^{\otimes 2n} \tilde{\Pi} W^{\otimes 2n }_{\vec{\bt}}\right] = \sum_{\vec{k}, \vec{k}'} c(\vec{k}, \vec{k}') B_{\vec{k}} \otimes B_{\vec{k}'}\,,
    \label{eq:2n_fold_expectation}
\end{equation}
where $c(\vec{k}, \vec{k}')$ builds a matrix $C$ of size $d_{\rm PI} \times d_{\rm PI}$. This matrix $C$ fully characterizes the symmetrized PI measurement channel $\MC_{\rm symm-PI}$ leading to the expression 
\begin{equation}
    \MC_{\rm symm-PI}(\rho) = \sum_{\vec{k}, \vec{k}'}c(\vec{k}, \vec{k}') \Tr[\rho B_{\vec{k}}]B_{\vec{k}'}\;.
    \label{eq:m_symmetrized}
\end{equation}

For a given linear operator $A \in \LC(\HC)$, the action of the symmetrized PI measurement channel admits a compact vectorized representation, 
\begin{equation}
    |\MC_{\rm symm-PI}(A)\rangle\!\rangle = C |A^{S_n}\rangle\!\rangle\;,  
\end{equation}
where $A^{S_n} = \TC_{S_n}(A)$ denotes the projection of $A$ onto the symmetric subspace ${\rm comm}(S_n)$. The vector $|A^{S_n}\rangle\!\rangle$ is a $d_{\rm PI}$ dimensional vector defined via its components in the symmetrized Pauli basis as $|A^{S_n}\rangle\!\rangle_\al = \sum_{\al \in [d_{\rm PI}]}\Tr[B_\al\ad A ]|B_\al\rangle\!\rangle$. Here, we relabel the multi-index $\vec{k}$ by a single index $\al \in [d_{\rm PI}]$ for notational simplicity. 

Since the measurement channel is supported on ${\rm comm}(S_n)$, its (pseudo-)inverse $\MC_{\rm symm-PI}^{-1}$ is obtained by inverting the matrix $C$ on this subspace. Explicitly, the inverse channel acts as $|\MC_{\rm symm-PI}(A)\rangle\!\rangle = C^{-1} |A^{S_n} \rangle\!\rangle$. 
Combining the expression with Eq.~\eqref{eq:o_estimate}, the classical shadow estimator for the expectation value of an observable $O$ can be written in vectorized form as:
\begin{equation}
    \hat{o} = \Tr[\hat{\rho} O] = \langle\!\langle (W\ad_{\vec{\bt}})^{\otimes n}  \Pi_h W_{\vec{\bt}}^{\otimes n}| C^{-1} |O\rangle\!\rangle
    \label{eq:cs_estimate_vector}
\end{equation}
with the vector 
\begin{equation}
    |(W\ad_{\vec{\bt}})^{\otimes n}  \Pi_h W_{\vec{\bt}}^{\otimes n} \rangle\!\rangle = \sum_{\vec{k}} \Tr[B_{\vec{k}}(W\ad_{\vec{\bt}})^{\otimes n}  \Pi_h W_{\vec{\bt}}^{\otimes n}] |B_{\vec{k}}\rangle \! \rangle \;. 
\end{equation}
Therefore, evaluating the symmetrized PI-CS estimator reduces to two tasks: (i) computing matrix $C$ that characterizes the measurement channel, and (ii) evaluating the entries of the measurement vector $|(W\ad_{\vec{\bt}})^{\otimes n}  \Pi_h W_{\vec{\bt}}^{\otimes n} \rangle\!\rangle$ in the symmetrized Pauli basis.

We begin by computing the matrix $C$. To this end, it is convenient to first express the Hamming weight projectors $\{\Pi_h\}$ in the symmetrized Pauli basis. We introduce the shorthand notation= $B^m = B_{(0, 0, m, n-m)} = \TC_{S_n}(Z^{\otimes m} \eye^{\otimes (n -m)})$ which denotes the normalized symmetrization of Pauli-$Z$ operators acting on $m$ qubits. In this basis, each Hamming-weight projector with weight $h$ admits the form
\begin{equation}
    \Pi_h = \sum_{m = 0}^n \al(h, m) B^m\;.
    \label{eq:Hamming_weight_projector}
\end{equation}
Since $B^m$ are orthonormal, the coefficients $\al(h, m)$ can be evaluated as
\begin{equation}
    \al(h,m) = \Tr[B^m \Pi_h] = d^{-\frac{1}{2}}\binom{n}{m}^{\frac{1}{2}}a(h,m), ~~\text{with } a(h,m) = \sum_{l = 0}^h\binom{m}{l}\binom{n-m}{h -l}(-1)^{l}\;. 
\end{equation}

We now evaluate the expression of $2n$-qubit operator $\widetilde{\Pi}$ defined in Eq.~\eqref{eq:measurement_channel}. Substituting Eq.~\eqref{eq:Hamming_weight_projector} into its definition yields
\begin{equation}
    \tilde{\Pi} = \sum_{h = 0}^n \left(\Pi_h\right)^{\otimes 2} = \sum_{m, m'=0}^n \tilde{\al}(m, m') B^m \otimes B^{m'}, \qquad \tilde{\al}(m, m') = \sum_h \al(h, m) \al(h, m').
    \label{eq:pi_tilde}
\end{equation}
Inserting Eq.~\eqref{eq:pi_tilde} into the two-copy moment $\Ebb_{W}\left[ (W_{\vec{\bt}}\ad)^{\otimes 2n} \tilde{\Pi} W^{\otimes 2n }_{\vec{\bt}}\right]$ given by Eq.~\eqref{eq:2n_fold_expectation} and collecting all the expressions lead to a closed form expression for the matrix elements $C$ as:
\begin{equation}
    c(\vec{k}, \vec{k}') =
    \begin{cases}
        \frac{(-1)^{\frac{\abs{k - k'}}{2}}}{d(k + k' + 1)\sqrt{\vec{k}!\vec{k'}!(n-k)!(n-k')!}} \cdot \frac{(\vec{k} + \vec{k'})!}{\left(\frac{\vec{k} + \vec{k}'}{2}\right)!}\cdot \frac{(2n - k - k')!}{\left(\frac{2n - k - k'}{2}\right)!}  & \text{if } k_\al + k'_\al \in 2\mathbb{N},~~\forall \al \in \{X,Y,Z\}\,,  \\ 
        0 & \text{otherwise.}
    \end{cases} 
    \label{eq:ckk_expression}
\end{equation}
We again use the shorthand notation for multi-index factorials $\vec{k}! = k_X!\cdot k_Y! \cdot k_Z!$ and denote by $k := \abs{\vec{k}} = k_X + k_Y + k_Z$ the total weight of the vector $\vec{k}$. 
The parity constraint implies that $C$ decomposes into 8 independent blocks, corresponding to the parity sectors of $(k_X, k_Y, k_Z)$. Constructing $C$ explicitly requires $\order{n^6}$ operations. 

Having determined $C$, we now compute the entries of the measurement vector $|(W_{\vec{\bt}}\ad)^{\otimes n} \Pi_h W_{\vec{\bt}}\rangle\!\rangle$ on which the inverse matrix $C^{-1}$ acts following Eq.~\eqref{eq:cs_estimate_vector}. Using the expansion of $\Pi_h$ in Eq.~\eqref{eq:Hamming_weight_projector} together with the definitions of Pauli basis $B_{\vec{k}}$ (see Eq.~\eqref{eq:B_pauli}) and single-qubit rotation $W_{\vec{\theta}}$ (c.f. Eq.~\eqref{eq:W_theta}), we obtain
\begin{align}
    \Tr[B_{\vec{k}}(W_{\vec{\bt}}\ad )^{\otimes n} \Pi_h W_{\vec{\bt}}^{\otimes n}] & = \sum_{m = 0}^n \al(h, k) \Tr[B_{\vec{k}} (W_{\vec{\bt}}\ad)^{\otimes n}  B^m W_{\vec{\bt}}^{\otimes n}]  \nonumber \\ & = \frac{\sqrt{n!}}{\sqrt{d \vec{k}!(n-k)!}} a(h, k) \cos(\bt_1)^{k_X}\sin(\bt_1)^{k_Y} \sin(\bt_2)^{k_X + k_Y} \cos(\bt_2)^{k_Z}\;.
\end{align}
Together with the inverse matrix $C^{-1}$, these expressions fully determining the symmetrized PI-CS estimator.

The main computational bottleneck of evaluating symmetrized PI-CS estimators is the inversion of the matrix $C$ given in Eq.~\eqref{eq:cs_estimates}.  By reordering the vector $\vec{k}$, we can reduce the matrix $C$ into 8 blocks, thereby reducing the effective size of the matrices that must be handled. Nevertheless, a direct inversion of these blocks remains computationally expensive.
In practice, we therefore avoid explicitly computing $C^{-1}$, which costs $\order{d_{\rm PI}^3} \approx \order{n^9}$. Instead, we perform a one-time LU decomposition of $C$ and subsequently solve the system of linear equations $C\kett{x} = |O\rangle\!\rangle$ for each observable $O$ of interest. The following theorem summarizes the above discussion: 

\begin{theorem}[Cost of post-processing symmetrized PI-CS in symmetrized Pauli basis]
    Consider the symmetrized PI-CS applied on $S_n$-equivariant initial state. Under the assumption that matrix $C$ is precomputed and its LU decomposition is evaluated, the cost of post-processing the symmetrized PI-CS scales as $N_{\CT} \in \order{n^6}$.
\end{theorem}

Finally, the following bound can be found for the observable estimator: 
\begin{equation}
    \Var[\hat{o}] \le (2n + 1)\norm{O}_F^2. 
\end{equation}
where $\norm{A}_F$ denotes the Frobenius norm of the matrix $A$, determining the quantum sample complexity of symmetrized PI-CS protocol.

\subsection{Symmetrized PI-CS in Schur basis}
One issue with the symmetrized PI-CS in the symmetrized Pauli basis explained in the section above is that we cannot obtain $\hat{o}_\lm$, which is required in the classical simulation procedure explained in the main text. Therefore, it is needed to work in Schur basis.  
In this section, we formulate the symmetrized PI-CS protocol expressed in the Schur basis. Up to this point, we have interpreted the dimension label $q_\lm$ as the Hamming weight associated with the symmetric component of the canonical Schur basis. Here, we introduce a new label $\tilde{q}_\lm$, defined by
\begin{equation}
    \tilde{q}_\lm \equiv \tilde{q}_{\lm(m)} =  q_\lm - \left(\frac{n}{2} - m\right)\;, 
\end{equation}
which satisfies $\tilde{q}_\lm \in Q_\lm$ where $Q_\lm$ denotes the set of all admissible half-integer values of $\tilde{q}_\lm$ for the given irrep label $\lm$ as 
\begin{equation}
    Q_\lm = \{-s_\lm, -s_\lm + 1,\dots,s_\lm -1, s_\lm\}\;, \qquad s_\lm \equiv s_{\lm(m) } = \frac{n}{2} - m\;, 
\end{equation}
with $|Q_\lm| = d_\lm = n - 2m + 1 $.  With this convention, the Schur basis states are relabeled as $\ket{\lm, p_\lm, q_\lm} \equiv \ket{\lm, p_\lm, \tilde{q}_\lm}$. 

Using the relabeled Schur basis states, we begin by introducing an orthonormal PI operator Schur basis, defined as
\begin{equation}
    B_{\rm Schur} = \left\{B_{\tilde{q}_\lm, \tilde{q}'_\lm}^\lm   = \frac{1}{\sqrt{m_\lm}} \sum_{p_\lm =1}^{m_\lm} \ketbra{\lm, p_\lm, \tilde{q}_\lm}{ \lm, p_\lm, \tilde{q}'_\lm} \right\}_{\lm, \tilde{q}_\lm, \tilde{q}'_\lm }\;.
\end{equation}
To derive the explicit form of the classical shadow estimators, we now express Eq.~\eqref{eq:HW_projectors} in terms of the  Schur basis as 
\begin{equation}
    \Pi_h = \sum_{\lm: \tilde{q}(h) \in Q_\lm} B^{\lm}_{\tilde{q}(h), \tilde{q}(h)} \sqrt{m_\lm}\;, 
    \label{eq:Pi_h_Schur}
\end{equation}
where we define the notation $\tilde{q}(h) = h - \frac{n}{2}$ to make explicit the dependence of the label $\tilde{q}$ on the Hamming weight $h$. The condition $\lm: \tilde{q}(h) \in Q_\lm$ indicates that the summation runs over irrep label $\lm$ for which the value $\tilde{q}(h)$ belongs in the set $Q_\lm$. 
Having established this Schur-basis representation of $\Pi_h$, we now turn to constructing the associated $2n$-qubit operator $\widetilde{\Pi}$ appearing in Eq.~\eqref{eq:measurement_channel}.  Substituting Eq.~\eqref{eq:Pi_h_Schur}  into the definition of $\widetilde{\Pi}$, we obtain
\begin{equation}
    \widetilde{\Pi} = \sum_{h = 0}^n (\Pi_h)^{\otimes 2} = \sum_{\tilde{q} \in Q_{(n, 0)}}\sum_{\substack{\lm: \tilde{q}\in Q_\lm \\ \nu : \tilde{q}\in Q_\nu } } \sqrt{m_\lm m_\nu} \tilde{A}^{\lm, \nu}_{(\tilde{q}, \tilde{q}), (\tilde{q}, \tilde{q})}\;,
\end{equation}
where $Q_{(n, 0)}$ denotes the largest admissible set among the $Q_\lm$, corresponding to the irrep $\lm = (n, 0)$ with $s_{(n, 0)} = \tfrac{n}{2}$. Here, we have introduced the operators $\tilde{A}^{\lm, \nu}_{(\tilde{q}_\lm, \tilde{q}_\lm'), (\tilde{q}_\nu, \tilde{q}'_\nu)}$ as 
\begin{equation}
     \tilde{A}^{\lm, \nu}_{(\tilde{q}_\lm, \tilde{q}_\lm'), (\tilde{q}_\nu, \tilde{q}'_\nu)} = B^\lm_{\tilde{q}_\lm, \tilde{q}_\lm'} \otimes B^{\nu}_{\tilde{q}_\nu, \tilde{q}_\nu'}\;.
\end{equation}

We next evaluate the average $\Ebb_{W\sim \mathbb{SU}(2)}\left[ (W \ad)^{\otimes 2n} \tilde{\Pi} W^{\otimes 2n }\right] $, which corresponds to twirling $\widetilde{\Pi}$ with respect to $2n$-fold tensor products $W^{2n}$. Equivalently, it implements the orthogonal projection of $\widetilde{\Pi}$ onto the space of operators commuting with $W^{\otimes 2n}$ for all $W \sim \mathbb{SU}(2)$. This commutant space can be spanned by an orthogonal set of $2n$-qubit Schur operator basis 
\begin{equation}
    B^{2n} = \left\{ \widetilde{B}^{\mu}_{p_\mu, p_\mu'} = \frac{1}{\sqrt{d_\mu}} \sum_{\tilde{q}_\mu \in Q_\lm}| \mu, p_\mu, \tilde{q}_\mu \rangle\!\rangle\!\langle\!\langle \mu, p'_\mu, \tilde{q}_\mu| \right\}_{\mu, p_\mu, p'_\mu}\;.
\end{equation}
Therefore, the average $\Ebb_{W\sim \mathbb{SU}(2)}\left[ (W \ad)^{\otimes 2n} \tilde{\Pi} W^{\otimes 2n }\right] $ can be expanded in terms of those basis as : 
\begin{equation}
    \Ebb_{W\sim \mathbb{SU}(2)}\left[ (W \ad)^{\otimes 2n} \tilde{\Pi} W^{\otimes 2n }\right] = \sum_{p_\mu, p'_\mu} \Tr[\widetilde{B}^{\mu}_{p_\mu, p_\mu'} \widetilde{\Pi}] \widetilde{B}^{\mu}_{p_\mu, p_\mu'}\;.
\end{equation}

We now introduce the coefficients  $C^{\tilde{q}_\lm, \tilde{q}_\nu}_{\lm, \nu; \mu}$ defined in terms of the standard Clebsch--Gordan coefficients $CG(j_1, j_2, J, m_1, m_2, M)$ as 
\begin{equation}
    C^{\tilde{q}_\lm, \tilde{q}_\nu}_{\lm, \nu; \mu} = CG(s_\lm, s_\nu, s_\mu, \tilde{q}_\lm, \tilde{q}_\nu, \tilde{q}_\lm + \tilde{q}_\nu)\;. 
\end{equation}
This definition makes explicit the correspondence between the standard angular momentum quantum numbers $(j,m)$ and the Schur basis related symbols $(s_\lm, \tilde{q}_\lm)$. Using these coefficients, we evaluate the overlap with respect to $\widetilde{B}^{\mu}_{p_\mu, p_\mu'}$ and obtain the expansion 
\begin{equation}
    \Ebb_{W\sim \mathbb{SU}(2)}\left[ (W \ad)^{\otimes 2n} \tilde{\Pi} W^{\otimes 2n }\right] = \sum_{\lm, \nu} \sum_{\mu \doteq \lm + \nu}  \frac{\widetilde{\Pi}_{\mu, \lm, \al} \sqrt{m_\lm m_\nu}}{d_\mu} \left(\sum_{\substack{\tilde{q}_\lm, \tilde{q}_\nu, \tilde{q}_\lm', \tilde{q}_\nu' \\ \tilde{q}_\lm+ \tilde{q}_\nu = \tilde{q}'_\lm + \tilde{q}'_\nu \in Q_\mu}}C^{\tilde{q}_\lm', \tilde{q}_\nu'}_{\lm, \nu, \mu} C^{\tilde{q}_\lm, \tilde{q}_\nu}_{\lm, \nu; \mu} \widetilde{A}^{\lm, \nu}_{(\tilde{q}_\lm, \tilde{q}'_\lm), (\tilde{q}_\nu, \tilde{q}_\nu')}\right)\;.
    \label{eq:twirl_schur_basis}
\end{equation}
where the coefficient $\widetilde{\Pi}_{\mu, \lm, \nu}$ is given by projection of $\widetilde{\Pi}$ on the operator $\tilde{B}^{\mu}_{p_\mu p_\mu}$ as:
\begin{equation}
    \widetilde{\Pi}_{\mu, \lm, \nu} = \sqrt{m_\mu} \Tr[\widetilde{B}^{\mu}_{p_\mu, p_\mu} \widetilde{\Pi}] = \sum_{\substack{\tilde{q}: \\ \tilde{q}\in  Q_\lm \cap Q_\nu \\ 2\tilde{q}\in Q_\mu }} (C^{\tilde{q}, \tilde{q}}_{\lm, \nu, \mu})^2 \;. 
    \label{eq:coeff_pi}
\end{equation} 
In Eq.~\eqref{eq:twirl_schur_basis}, the summation constraint $\mu \doteq \lm + \al$ indicates that the sum runs over those irrep $\mu$  for which there exist $\tilde{q}_\lm\in Q_\lm$ and $\tilde{q}_\nu\in Q_\nu$ such that $\tilde{q}_\lm + \tilde{q}_\nu \in Q_\mu$. This condition encodes  the Clebsch--Gordan selection rules associated with angular momentum addition, and is equivalent to the triangle inequality between spin numbers $\abs{s_\lm - s_\nu} \le s_\mu \le s_\lm + s_\nu$.

Inserting Eq.~\eqref{eq:twirl_schur_basis} to Eq.~\eqref{eq:measurement_channel}, we obtain the following expression for the measurement channel 
\begin{align}
    \MC_{\rm symm-PI}(\rho)  = \sum_\nu \sum_{\tilde{q}_\nu, \tilde{q}_\nu'}\left(\sum_{\lm} \sum_{\tilde{q}_\lm, \tilde{q}_\lm'}  c(\lm, \tilde{q}_\lm, \tilde{q}_\lm', \nu, \tilde{q}_\nu, \tilde{q}_\nu')  \Tr[\rho B^\lm_{\tilde{q}_\lm, \tilde{q}_\lm'}] \right) B^{\nu}_{\tilde{q}_\nu, \tilde{q}_\nu'} 
    \label{eq:measurement_channel_schur_basis}
\end{align}
where the coefficient $ c(\lm, \tilde{q}_\lm, \tilde{q}_\lm', \nu, \tilde{q}_\nu, \tilde{q}_\nu') $ is defined as
\begin{equation}
     c(\lm, \tilde{q}_\lm, \tilde{q}_\lm', \nu, \tilde{q}_\nu, \tilde{q}_\nu') = \eye_{[\tilde{q}_\lm + \tilde{q}_\nu = \tilde{q}_\lm' + \tilde{q}_\nu' ]} \sum_{\substack{\mu \\ \mu \doteq \lm + \nu, \\ \tilde{q}_\lm + \tilde{q}_\nu \in Q_\mu}} \frac{\widetilde{\Pi}_{\mu, \lm, \nu} \sqrt{m_\lm m_\nu}}{d_\mu} \left(C^{\tilde{q}_\lm', \tilde{q}_\nu'}_{\lm, \nu;\mu} C^{\tilde{q}_\lm, \tilde{q}_\nu}_{\lm, \nu; \mu}\right)\;.
\end{equation}
Here, $\eye_{[\textrm{C}]}$ denotes a function such that $\eye_{[\textrm{C}]} = 1$ if the condition $\rm C$ is satisfied and $0$ otherwise. 

To further simplify the computation, we exploit the additional block structure of the measurement channel. To this end, we relabel the Schur operator basis $B^{\lm}_{\tilde{q}_\lm, \tilde{q}_\lm'}$ as $B^\lm_{\Delta_\lm, \tilde{q}_\lm + \Delta_\lm}$ according to their shift $\Delta_\lm = \tilde{q}_\lm - \tilde{q}_\lm'$. Under this notation, the measurement channel $\MC_{\rm symm-PI}$ takes the form: 
\begin{equation}
\MC_{\rm symm-PI}(\rho) = \sum_\nu \sum_{\tilde{q}_\nu, \Delta_\nu} \left(  \sum_{\lm}  \sum_{\tilde{q}_\lm, \Delta_\lm} \tilde{c}(B^\l_{\Delta_\lm, \tilde{q}_\lm}, B^\nu_{\Delta_\nu,  \tilde{q}_\nu}) \Tr[\rho B^\l_{\Delta_\lm, \tilde{q}_\lm}] \right)  B^\nu_{\Delta_\nu,  \tilde{q}_\nu}\;, 
\label{eq:measurement_channel_relabeld}
\end{equation}
with the coefficients defined as 
\begin{align}
    \tilde{c}(B^\l_{\Delta_\lm, \tilde{q}_\lm}, B^\nu_{\Delta_\nu,  \tilde{q}_\nu}) & \equiv c(\lm, \tilde{q}_\lm, \tilde{q}_\lm + \Delta_\lm, \nu, \tilde{q}_\nu, \tilde{q}_\nu + \Delta_\nu) \nonumber \\ &  = \eye_{[\Delta_\lm = -\Delta_\nu]} \sum_{\substack{\mu \\ \mu \doteq \lm + \nu, \\ \tilde{q}_\lm + \tilde{q}_\nu \in Q_\mu}} \frac{\widetilde{\Pi}_{\mu, \lm, \nu} \sqrt{m_\lm m_\nu}}{d_\mu} \left(C^{\tilde{q}_\lm + \Delta_\lm, \tilde{q}_\nu+ \Delta_\nu }_{\lm, \nu;\mu} C^{\tilde{q}_\lm, \tilde{q}_\nu }_{\lm, \nu; \mu}\right)  \;. 
    \label{eq:tilde_c}
\end{align}
By construction, the measurement channel has non-trivial contribution only when $\Delta_\lm = -\Delta_\nu$. As a result, by reordering the basis operator according to the shift label $\Delta$, the matrix representation $C$ of the measurement channel can be decomposed into  a direct sum of independent blocks indexed by $(\Delta, -\Delta)$. As $\Delta \in \{-n, -n + 1, \dots, n \}$, the channel decomposes into $2n + 1$ blocks, each acting on the subspace spanned by all basis operators satisfying $\Delta_\lm = \Delta$. This block structure substantially reduces the computational complexity required for inverting $C$. 

Finally, we compute the entries in the vector $|(W\ad)^{\otimes n } \Pi_h W^{\otimes n} \rangle \! \rangle$ on which we apply the inverse measurement channel with $W = W_{\vec{\bt}}$ defined in Eq.~\eqref{eq:W_theta}. Recall that, under vectorization of operators, the adjoint action of an operator $A$ on an arbitrary operator $B$, $B \mapsto A\ad B A$, is represented  by a matrix $\bar{A}$ acting on the vectorized form $|B\rangle \! \rangle$. We denote $\bar{Y}$ and $\bar{Z}$ to be the matrix form of the commutation action generated by $\sum_i Y_i$ and $\sum_i Z_i$, respectively. With this notation, we obtain 
\begin{equation}
    |(W_{\vec{\theta}}\ad)^{\otimes n } \Pi_h W_{\vec{\theta}}^{\otimes n} \rangle \! \rangle = \exp\left(-i \frac{\theta_3}{2} \bar{Z}\right) \exp\left(-i \frac{\theta_2}{2} \bar{Y}\right) \exp\left(-i \frac{\theta_1}{2} \bar{Z}\right)|\Pi_h \rangle\!\rangle \;, 
\end{equation}
where the entries of vector $ |\Pi_h \rangle\!\rangle$ can be evaluated using Eq.~\eqref{eq:Pi_h_Schur}.

Combining all the results derived above, we can now compute the classical shadow estimator associated with each irrep. To this end, we define $\hat{o}_\lm = \Tr[\hat{\rho}_\lm O_\lm]$ where $\hat{\rho}_\lm$ denotes the estimator of $\rho_\lm$, and rewrite Eq.~\eqref{eq:o_estimate} as a sum of $\hat{o}_\lm$ from each irrep $\lm$ as 
\begin{equation}
    \hat{o} = \Tr[\hat{\rho} O] = \sum_\lm \Tr[\hat{\rho} (\eye_{m_\lm} \otimes O_\lm)] = \sum_\lm \langle\!\langle (W_{\vec{\theta}}\ad)^{\otimes n } \Pi_h W_{\vec{\theta}}^{\otimes n}  |   C^{-1}|O_\lm\rangle\!\rangle =: \sum_\lm m_\lm \hat{o}_\lm
\end{equation}
where the matrix $C$ is block-diagonalized and the vector $|O_\lm \rangle\!\rangle $ has non-zero components only in the entries corresponding to the operator $B^\lm_{\tilde{q}_\lm, \tilde{q}'_\lm}$. We emphasize that, although $C$ is block-diagonal, $|O_\lm\rangle\!\rangle$ has support across multiple blocks of $C$, as each block is indexed solely by the shift values $(\Delta, -\Delta)$ which is independent of $\lm$.

As the above analysis shows, the computation of the classical shadow estimators requires evaluating a large number of Clebsch--Gordan coefficients to compute the CS estimators, which constitutes the main computational bottleneck of the protocol. Assuming that the cost of evaluating a single Clebsch--Gordan coefficient is $\order{1}$ using a precomputed lookup table, the total computational cost scales at least as $\order{n^6}$, since the construction of $\MC_{\rm symm-PI}$ involves nested loops over the indices $\lm$,  $\nu$, $\mu$, $q_\lm$, $q_\nu$ and $\Delta$ (see Eqs.~\eqref{eq:measurement_channel_relabeld} and \eqref{eq:tilde_c}), each of which takes at most $\order{n}$ distinct values. As a result, the procedure becomes computationally demanding for large numbers of qubits $n$. Without pre-computation, an additional overhead is incurred due to the repeated evaluation of the Clebsch--Gordan coefficients.  

This cost can be amortized, since the measurement channel matrix $C$ depends only on the measurement protocol and the number of qubits $n$, and is independent of the initial state. Hence, for a fixed $n$, the matrix $C$  needs to be precomputed only once and can be reused across multiple CS simulations. In practice, this significantly reduces the per-simulation computational overhead and makes the approach feasible even for moderately large $n$.

\end{document}